\documentclass[mathpazo]{cicp}
\usepackage{amsmath}
\usepackage{amssymb}
\usepackage{amsthm}
\usepackage{amsfonts}
\usepackage{graphicx}
\usepackage{subfigure}
\usepackage[usenames]{color}
\usepackage{ulem}
\usepackage{bm}
\usepackage{microtype}
\usepackage[colorlinks=true]{hyperref}

\begin{document}
\title[Numerical Continuation of resonances]{Numerical Continuation of resonances and bound states in coupled channel Schr\"odinger equations}

\author[P.\ Klosiewicz et.\ al.]{Przemys\l{}aw K\l{}osiewicz\corrauth 
, Jan Broeckhove and Wim Vanroose}
\address{Department of Mathematics and Computer Science, Universiteit Antwerpen,\\ Middelheimlaan 1, B-2020 Antwerpen}
\email{{\tt przemyslaw.klosiewicz@ua.ac.be} (P.\ K\l{}osiewicz)}


\begin{abstract}
In this contribution, we introduce numerical continuation methods and bifurcation theory, techniques which find their roots in the study of dynamical systems, to the problem of tracing the parameter dependence
of bound and resonant states of the quantum mechanical Schr\"o\-din\-ger equation. We extend previous work on the subject \cite{Broeckhove2009} to systems of coupled equations.

Bound and resonant states of the Schr\"o\-din\-ger equation can be determined through the poles of the $S$-matrix, a quantity that can be derived from the asymptotic form of the wave function. We introduce a regularization procedure that essentially  transforms the $S$-matrix into its inverse and improves its smoothness properties, thus making it amenable to numerical continuation. This allows us to automate the process of tracking bound and resonant states when parameters in the Schr\"o\-din\-ger equation are varied. We have applied this approach to a number of model problems with satisfying results.
\end{abstract}

\pac{03.65.Ge, 03.65.Nk, 82.20.Xr, 47.20.Ky}

\maketitle

%
%
\section{Introduction}\label{sec:Introduction}
The appearance of resonances is of ever-growing interest in the study of wave phenomena as they are considered among the most 
important features of systems described by wave equations. They appear in systems that are penetrable by an impacting wave. Such systems 
allow the interior field to couple to the external domain which leaves a characteristic fingerprint on the far-field pattern of the scattered wave. 
Many examples appear naturally in acoustic scattering~\cite{fahy2007sound} and in fluid-mechanical structure interaction~\cite{dhia2007resonances}. 
In all cases the appearance of resonant states has a profound and important influence on the system's dynamics.

In the context of quantum mechanical scattering, resonant behavior also strongly influences the interactions between microscopic particles, which in turn 
has its influence on the reactivity of molecules and atoms described by such quantum mechanical models.  In molecular systems these resonances can 
easily turn into bound states if the molecular configuration changes. 

In~\cite{Broeckhove2009} we developed a framework for applying numerical continuation techniques in the context of bound states and resonances 
in spherically symmetric short-range potentials. We have shown that numerical continuation methods, originally developed in the study of dynamical 
systems, can be applied successfully to track bound and resonant states efficiently in terms of a varying system parameter. Moreover, this technique 
can be used to reveal subtle and interesting transitions and connections between states automatically.

The present work focuses on the extension of that procedure to coupled channel short-range systems. This extension is a logical step towards automated, 
efficient and robust methods for the study of interactions in scattering experiments. In all generality, these techniques can be applied to systems of 
coupled Helmholtz equations with variable wave numbers, as long as the short-range conditions are met.

The outline of the paper is as follows. Section~\ref{sec:Scattering} sets the coupled channel Schr\"o\-din\-ger equations in the context of non spherically 
symmetric quantum mechanical problems. In section~\ref{sec:Regularization} a regularization procedure is discussed that allows application of numerical 
continuation even though the underlying functions that characterize resonances and bound states are numerically and analytically not very well-suited. 
Section~\ref{sec:NumCont} provides a brief overview of basic numerical continuation methods and gives some pointers on the available implementations. 
Finally, in section~\ref{sec:results} we present several results obtained with our implementation of the discussed methods.

%
%
\section{Quantum scattering in coupled channel problems}
\label{sec:Scattering}

%
In this paper we discuss a coupled channel problem that derives from a one particle Schr\"odinger equation with a non-spherical potential
\begin{equation}
	\label{eq:3dschrodinger}
  	\left(-\frac{1}{2\mu}{\Delta} + V(\mathbold{r},\lambda) - E  \right) \psi(\mathbold{r}) =0,  
\end{equation}
where $\Delta$ is the three-dimensional Laplacian, $V(\mathbold{r},\lambda)$ is a potential with a system parameter $\lambda$ and $E$ is the complex-valued energy of the system.
The problem is such that for all $\mathbold{r}$ outside a bounded domain $V(\mathbold{r},\lambda) \approx 0$, i.e.\ the potential becomes negligibly small. Formally, the limitation to potentials that are negligible outside a certain radius is termed as the restriction to so called short-range potentials: $V(\mathbold{r},\lambda)$ must decay faster than $r^{-3}$ as $r=|\mathbold{r}|\to\infty$ and must be less singular than $r^{-2}$ in the origin $r=|\mathbold{r}|\to0$~\cite{Taylor2006}. The long-range Coulomb interaction requires a significantly different approach and is not discussed here. The boundary conditions that are appropriate in the short-range case require the solution to be zero at the origin and force the solution to tend to a linear combination of free waves (i.e.\ solutions of \eqref{eq:3dschrodinger} for $V=0$) at infinity~\cite{Taylor2006}.

In this section we briefly introduce elements from quantum scattering and partial wave expansion to arrive at the concept of resonant states. In the subsequent sections we then focus on applying numerical continuation to study the dependence of such resonances on the system parameter $\lambda$ in \eqref{eq:3dschrodinger}. For this reason, we will faithfully record the $\lambda$ dependence in our notations, even if it is at times somewhat cumbersome.

Equation \eqref{eq:3dschrodinger} can be written in spherical coordinates $(r,\theta, \varphi)$ around the center of the system. The differential operator $\Delta$ then splits into angular and radial differential operators and the solution can be expanded as a sum
\begin{equation}
	\psi(\mathbold{r}) = \psi(r,\theta,\varphi) = \sum_{l=0}^{\infty}\sum_{m=-l}^{l} \psi_{lm}(r) Y_{l}^{m}(\theta, \varphi),
	\label{eq:PartialWave}
\end{equation}
where $Y_{l}^{m}(\theta,\varphi)$ are the spherical harmonics, eigenfunctions of the angular differential operators in the Laplacian.  
In physics this decomposition is commonly referred to as the partial wave expansion. The labels $(l,m)$ are intimately connected to the irreducible representations of the rotational symmetry groups $SO(3) \supset SO(2)$ of equation \eqref{eq:3dschrodinger}. Each $(l,m)$-component of \eqref{eq:PartialWave} is known as a partial wave and the system is said to be modeled by multiple partial wave channels. We are now interested in localizing the resonances in such systems.

\subsection{Single-channel scattering}
\label{sphericalsymmetry}
To introduce some notations and concepts we first briefly look at a spherically symmetric potential, i.e.\ $V(\mathbold{r},\lambda) = V(r,\lambda)$. In this case, the partial wave channels are decoupled and the radial wave function $\psi_{lm}$ is identical for all $m$. Hence we can drop the index $m$ and equation \eqref{eq:3dschrodinger} turns into:
\begin{equation}
	\label{eq:equationspherical}
	\left(-\frac{1}{2\mu} \frac{d^2}{d r^2} + \frac{l(l+1)}{2\mu r^2} - E \right) \psi_{l}(r,\lambda) + V(r,\lambda) \psi_{l}(r, \lambda) = 0,
\end{equation}
for each $l$. The boundary condition at $r=0$ specifies $\psi_{l}(r,\lambda) = 0$ and at large $r$ the solution must be a linear combination of the spherical Riccati-Hankel functions (see appendix \ref{app:spherical}), the free incoming and outgoing waves:
\begin{equation}
	\label{eq:sphericalboundary}
	\psi_{l}(r, \lambda) \xrightarrow{r\to\infty} \frac{i}{2} \left( \hat{h}^{-}_{l}(kr) - \hat{h}^{+}_{l}(kr) S_{l}(k, \lambda) \right),
\end{equation}
The spherical Riccati-Hankel functions are the solutions obtained in the absence of the potential term in \eqref{eq:equationspherical}, i.e.\ when the equation reduces to a Helmholtz equation with a constant wave number $k=\sqrt{2\mu E}$. The $S_{l}(k, \lambda)$ is called the $S$-matrix for channel $l$ and it determines the scattering properties associated with potential $V(\mathbold{r},\lambda)$.

\subsection{Multi-channel scattering}
For a non-spherical potential the partial wave channels do not decouple. The wave function must be represented by a sum as in equation \eqref{eq:PartialWave}, which is typically truncated at an $l_\text{max}$. In most cases of physical interest the potential still has axial symmetry. It is well known that as a consequence the channels with different $m$ are decoupled and the solutions can be represented by
\begin{equation}
	\psi(\mathbold{r}) = \sum_{l=0}^{l_\text{max}} \psi_{lm}(r) Y_{l}^{m}(\theta, \varphi).
	\label{eq:PartialWaveAxial0}
\end{equation}
Upon substitution in \eqref{eq:3dschrodinger}, this generates a separate set of equations for each m. For the purpose of our exposition we may, without loss of generality, take $m=0$ and drop the index $m$ in the notation for the radial wave function 
\begin{equation}
	\psi(\mathbold{r}) = \sum_{l=0}^{l_\text{max}} \psi_{l}(r) Y_{l}^{0}(\theta, \varphi).
	\label{eq:PartialWaveAxial}
\end{equation}
This significantly simplifies notations further on in the discussion. Lifting the restriction of axial symmetry presents no fundamental difficulty: the number of channels increases  and summations over $m$ have to be reintroduced. We will however consider only the axially symmetric case with $m=0$ as it provides a sufficient basis to demonstrate that numerical continuation can be used to track resonant states in multi-channel systems.

Insertion of \eqref{eq:PartialWaveAxial} in \eqref{eq:3dschrodinger}, followed by a projection onto each of the partial wave channels (by multiplying with $Y_{l'}^{0}$, integrating over the angular variables and using the orthonormality of the spherical harmonics) leads to a set of coupled radial equations 
\begin{equation}
	\label{eq:coupledchannel}
	\left(-\frac{1}{2\mu} \frac{d^2}{d r^2} + \frac{l(l+1)}{2\mu r^2} - E \right) \psi_{l}(r, \lambda) \ + \ \sum_{l'=0}^{l_\text{max}} V_{ll'}(r,\lambda) \psi_{l'}(r, \lambda) \ = \ 0,
\end{equation}
for each channel $l$ and where
\begin{equation}
	V_{ll'}(r,\lambda) = \int \sin\theta \, \text{d}\theta \, \text{d}\varphi \, {Y_{l}^{0}}^{*}(\theta, \varphi) V(r,\theta,\lambda) Y_{l'}^{0}(\theta,\varphi),
\end{equation}
and where $^{*}$ denotes the complex conjugate. For every partial wave $\psi_{l}$ the boundary condition at $r=0$ is still $\psi_{l}(r, \lambda) = 0$ while for large $r$ one formulates $l_{\text{max}}+1$ independent boundary conditions (one for each $l'$) 
\begin{equation}
	\label{eq:incomingpartial}
	\psi^{l'}_{l}(r, \lambda) \xrightarrow{r\to\infty} \frac{i}{2} \left( \hat{h}_{l}^{-}(kr) \delta_{ll'} - \hat{h}_{l}^{+}(kr) S_{ll'}(k, \lambda) \right),
\end{equation}
where $\delta_{ll'}$ is the usual Kronecker symbol.

Each boundary condition corresponds to an incoming wave in a different channel $l'$ and gives rise to an independent associated solution with its channel components labeled by $l$. A key difference with the spherical case is that, due to the non-spherical action of the potential, outgoing wave components appear in channels other than the incoming channel, as indicated by the $S$-matrix defined by its elements $S_{ll'}(k, \lambda)$.

We now rewrite the above in terms of matrices by introducing the square matrices $\mathbold{T}$, $\mathbold{V}$ and the identity matrix $\mathbold{I}$ with dimensions equal to the number of channels:
\begin{equation}
	\left(\mathbold{T}(r)\right)_{ll'} =  \delta_{ll'}\: \left(-\frac{1}{2\mu} \frac{d^2}{d r^2} + \frac{l(l+1)}{2\mu r^2} \right),\qquad \left(\mathbold{V}(r,\lambda)\right)_{ll'} = V_{ll'}(r,\lambda).
\end{equation}
Eq.\ \eqref{eq:coupledchannel} then is equivalent to
\begin{equation}
	\left( \mathbold{T}(r) - E\mathbold{I} + \mathbold{V}(r, \lambda) \right) \mathbold{\Psi}^{l'}(r, \lambda) = 0,
\end{equation}
where the wave function is a vector function with dimension equal to the number of channels, corresponding to a solution with \eqref{eq:incomingpartial} (easily rewitten in vector format). In multi-channel calculations it is however customary to introduce the square matrix $\mathbold{\Psi}$ whose columns are the $\mathbold{\Psi}^{l'} $ (corresponding to \eqref{eq:incomingpartial}) for each of the $l'$. This then leads to the set of coupled channel equations and boundary conditions that determine the full $S$-matrix which is the focus of our study
\begin{equation}
	\label{eq:matrixeq}
	\left( \mathbold{T}(r) - E\mathbold{I} + \mathbold{V}(r, \lambda)  \right) \mathbold{\Psi}(r, \lambda) = 0,
\end{equation}
\begin{equation}
	\label{eq:matrixboundary1}
	\mathbold{\Psi}(0, \lambda) = 0,
\end{equation}
\begin{equation}
	\label{eq:matrixboundary2}
	\mathbold{\Psi}(r, \lambda) \xrightarrow{r\to\infty} \frac{i}{2}  \left( \hat{\mathbold{h}}^-(kr) - \hat{\mathbold{h}}^+(kr) \mathbold{S}(k,\lambda) \right).
\end{equation}
The matrices $\hat{\mathbold{h}}^{\pm}(kr)$ are diagonal matrices with the incoming and outgoing spherical Riccati-Hankel functions on the diagonal
\begin{equation}
	\left(\hat{\mathbold{h}}^{-}(kr)\right)_{ll'} =  \delta_{ll'} \hat{h}_{l}^{-}(kr), \qquad \left(\hat{\mathbold{h}}^{+}(kr)\right)_{ll'} = \delta_{ll'} \hat{h}_{l}^{+}(kr).
\end{equation}
The multi-channel $S$-matrix $\mathbold{S}$ is the object of interest that we will use to find the resonant states. We can extract the $S$-matrix from the numerical solution of \eqref{eq:matrixeq} with its boundary conditions by evaluating the Wronskian in a point $r$ outside the range of the potential 
\begin{equation}
	\label{eq:Smatrix_and_wronskian}
 	\mathbold{S}(k,\lambda) = \mathcal{W}\left[\mathbold{\Psi}(r,\lambda),\hat{\mathbold{h}}^{+}(kr)\right]^{-1}\,\, \mathcal{W}\left[\mathbold{\Psi}(r,\lambda),\hat{\mathbold{h}}^{-}(kr)\right], 
\end{equation}
where $\mathcal{W}$ denotes the Wronskian of two matrices defined as
\begin{equation}
	\label{eq:Wronskian}
	\mathcal{W}[\mathbold{A},\mathbold{B}] = \mathbold{A} \left(\frac{d \mathbold{B}}{dr} \right)-\left(\frac{d \mathbold{A}}{dr} \right)\mathbold{B},
\end{equation}
with $\mathbold{A}$, $\mathbold{B} \in \mathbb{C}^{n\times n}$.

%
%
\section{Resonances and regularization}
\label{sec:Regularization}

The aim of this article is to develop a robust method to trace the parameter dependence of resonances that appear in multichannel problems.  We define a resonance or bound state to be the solution of equation \eqref{eq:3dschrodinger} for an energy $E\in\mathbb{C}$ for which the $S$-matrix has 
a pole. This definition is widely used and it is motivated, amongst others, in \cite{Newton1982,Taylor2006}. 

It is natural to locate the resonances numerically by searching for the zeros of $\mathbold{S}(k,\lambda)^{-1}$. 
However, it is not the best numerical strategy to search for these zeros directly. 
A key property of the $S$-matrix is the symmetry $\mathbold{S}(k^{*},\lambda) = \mathbold{S}(k,\lambda)^{-1}$, where $k^{*}$ denotes 
the complex conjugate of $k$. This makes $\mathbold{S}(k,\lambda)$ unitary for $k\in\mathbb{R}$. 
As a result a pole of the $S$-matrix at $k$ is always accompanied by a zero of $\mathbold{S}(k,\lambda)$ at $k^*$. 
This makes the search for the zeros of $\mathbold{S}(k,\lambda)^{-1}$ a numerically hard problem, particularly when $k \ll 1$. In that 
case $\mathbold{S}(k,\lambda)$ has zeros and poles close to each other and it does not satisfy the smoothness 
conditions required to ensure rapid convergence if, for instance, Newton iterations are used to find the zeros.  

For single-channel radial problems, as in section \ref{sphericalsymmetry}, it was proposed in \cite{Broeckhove2009} to use a regularized function rather than 
the $S$-matrix itself to converge and to track down bound states and resonances. This function was defined as, per $l$-channel,
\begin{equation}
	\label{eq:one_channel_regularized}
	F_{l}(k,\lambda) = \frac{k^{2l+1}}{S_l(k,\lambda)-1},
\end{equation}
with $S_{l}$ as in \eqref{eq:sphericalboundary}. It was shown that for $k\neq 0 $ this function has a zero iff $S_l(k,\lambda)$ has a pole. In addition this function is bounded 
as long as $\int_0^\infty dr\, \hat{j}_l(kr) V(r,\lambda) \psi_l(k,r) \neq 0$, which does not hold only for very specific potentials. 

In this paper we extend this result to the more general class of coupled channel problems described in the previous section. We introduce, in extension of \eqref{eq:one_channel_regularized}, 
the $(l_\text{max} + 1)$ by $(l_\text{max} + 1)$ square matrix function
\begin{equation}
	\label{eq:multi_channel_regularized}
	\mathbold{F}(k,\lambda) = \mathcal{K}
	\left( \mathbold{S}(k,\lambda) - \mathbold{I} \right)^{-1},
\end{equation}
where $\mathcal{K}$ is a diagonal matrix with elements $(\mathcal{K})_{ii} = k^{2l_{i}+1}$ on the diagonal and where we extract the $S$-matrix using \eqref{eq:Smatrix_and_wronskian} from the numerical solution of \eqref{eq:matrixeq}, \eqref{eq:matrixboundary1} and \eqref{eq:matrixboundary2}. The function used for continuation then is
\begin{equation}
	\label{eq:det_multi_channel_regularized}
	\det(\mathbold{F}(k,\lambda)) = \left( \prod_{l=0}^{l_{\text{max}}} k^{2l+1} \right) \bigg/ \det(\mathbold{S}(k,\lambda)-\mathbold{I}).
\end{equation}

In appendix \ref{app:regularization} we show that $\det(\mathbold{F}(k,\lambda))$ has a zero if and only if $\det(\mathbold{S}(k,\lambda)^{-1})$ has a zero. This means that we can look for the bound states and resonances via \eqref{eq:det_multi_channel_regularized} and investigate the application of continuation to equation $\det(\mathbold{F}(k,\lambda))=0$ . In addition, the coalescing of zeros and poles that occurs in the $S$-matrix is now absent. Unfortunately, we cannot derive a bound for $\det(\mathbold{F}(k,\lambda))$ itself as was possible for \eqref{eq:one_channel_regularized}. See appendix \ref{app:regularization} for details.

%
%

\section{Numerical Continuation methods}
\label{sec:NumCont}
In the study of dynamical systems one is often interested in the dependence and evolution of solutions of nonlinear 
problems in terms of some system parameter. Generally these  solution sets can only be approximated numerically and finding 
those requires efficient and robust techniques. More specifically, numerical continuation techniques are concerned 
with approximating the solution set $\left\{ \mathbold{x} | F(\mathbold{x}) = \mathbold{0} \right\}$ of an underdetermined system of nonlinear equations:
\begin{align}
	F : \mathbb{R}^{n+1} &\longrightarrow \mathbb{R}^{n} \\
	\mathbold{x} = (\mathbold{u},\lambda) &\longmapsto F(\mathbold{u},\lambda),
\end{align}
that depends on some system parameter $\lambda$. Such equations arise for instance in the study of dynamical systems 
or parametrized systems of ODEs but virtually every problem that can be written in the above formulation and that satisfies 
certain smoothness conditions can be subjected to those methods.

Due to the computationally intensive nature of many of these problems, efficiency is one of the main concerns in the numerical study of 
such dynamical systems. The number of evaluations of the function $F(\mathbold{x})$ should therefore be kept to a minimum. Additionally, 
solution sets often exhibit complex geometries with intersecting and bifurcating solution branches. Generally, the study of bifurcating 
branches involves rigorous stability analysis of the underlying solutions and is a complicated subject on its own. In this work the only 
type of bifurcations that can occur, are so called ``simple'' bifurcation points. These manifest themselves as two intersecting solution 
branches and are characterized by the dimension of the null-space of the Jacobian $F_{\mathbold{x}}(\mathbold{x}_{t})$ being 2 in 
a point $\mathbold{x}_t$, which is called a ``branch point''. In many problems, however, the bifurcations are much more involved and 
a thorough treatise on bifurcations and stability of solutions can be found in~\cite{Seydel1994,Doedel2007,Allgower1990,Keller1977}.

In the last few decades various techniques for numerical continuation have been developed. Many of those fall in the category of so 
called ``predictor-corrector'' methods: starting  from a known solution point $\mathbold{x}_{i}$ (such 
that $F(\mathbold{x}_{i})=\mathbold{0}$), the next one is found by first taking a guess based on the previous point and correcting 
that guess in a second step. As such, starting from one known initial point $\mathbold{x}_{0}$ in the solution set an approximation 
$\left\{ \mathbold{x}_{i} | i=0,\ldots,n \text{ with } F(\mathbold{x}_{i})=\mathbold{0} \right\}$ of that solution set is constructed by finding 
subsequent points on the curve.

Many choices can be made for the implementation of both steps. A good overview can be found in \cite{Seydel1994,Doedel2007,Allgower1990,Keller1977}. 
The popular ``pseudo arclength'' method \cite{Keller1969} takes a predictor-corrector approach. The predictor step uses the tangent $\mathbold{t_{i}}$ to the solution curve at $\mathbold{x}_{i}$ to construct 
a prediction $\mathbold{x}_{i+1}^{p}$ for the next point $\mathbold{x_{i+1}}$. The corrector step consists of Newton iterations on the system of equations $F = \mathbold{0}$ 
augmented with an additional equation that constrains the corrector step to a plane orthogonal to the tangent $\mathbold{t_{i}}$. The Newton iterations then converge to the intersection of that plane and the solution curve, guaranteeing a unique next point $\mathbold{x_{i+1}}$ in the continuation  \cite{Keller1969}.

The pseudo arclength algorithm has been implemented in the AUTO \cite{Doedel1981,AUTO2007} library which we have used in this work. 
Other well-known implementations of numerical continuation algorithms include LOCA~\cite{Salinger2002} which is a part of the Trilinos 
framework~\cite{Heroux2005}, \textsc{Matcont}~\cite{Dhooge2006} (a Matlab implementation) and Multifario~\cite{Henderson2002} 
which allows multi-pa\-ra\-me\-ter continuation.

%
%

\section{Numerical Results}
\label{sec:results}
The results of this section have been obtained with a renormalized Numerov solver \cite{Johnson1977,Johnson1978} for the calculation of the coupled channel solutions of equation \eqref{eq:matrixeq}. 
However, these results are independent of the particular numerical method used to solve the equations as long as the method is able to extract the $S$-matrix e.g. through equation \eqref{eq:Smatrix_and_wronskian}. 
Indeed, the regularized function $\det(\mathbold{F}(k,\lambda))$ is the only required input for the continuation routines.

To demonstrate our approach we have applied numerical continuation to a number of model problems involving well-behaved short-range potentials. The numerical implementation of the 
continuation algorithm  itself is provided by the AUTO-07p Fortran library \cite{AUTO2007}. The coupled channel renormalized Numerov solver and the code to obtain the regularized 
function $\det(\mathbold{F}(k,\lambda))$ are used as driver routines for AUTO and we have implemented them in C++.

\subsection{Comparison with analytically solvable models}
\label{sec:model1}
As a first testing example for the study of numerical code and regularization we have compared our results thoroughly with a model problem from \cite{Vanroose2002}. There, among 
other examples, a short-range two-channel coupling model for the low-energy electron-CO$_{2}$ scattering is studied. An analytically solvable expression with square-well potentials 
is presented together with resulting bound/resonant pole trajectories.

The application of numerical continuation on both the analytical expressions and the numerical solutions of this model matched the results previously published by the authors. 

\subsection{Gauss potential with $s$-wave and $p$-wave coupling}
\label{sec:model2}
In this section we look at a system of two coupled Gaussian potential wells with angular momenta $l_{1}=0$ (referred to as $s$-wave in physics) and $l_{2}=1$ (referred to as $p$-wave). The mass is chosen to be $\mu=1$.
The two channel potentials and the coupling potential are set to:
\[
	V_{ij}(r,\lambda) = -\lambda_{ij}e^{-r^{2}},
\]
where $i,j\in\{1,2\}$ and $V_{ij}(r,\lambda) = V_{ji}(r,\lambda)$.

Let us first consider each of the two channels separately, i.e.\ uncoupled ($\lambda_{12}=0$). Setting $\lambda_{11}$, the strength of the first-channel potential, to 7 gives a 
system with one bound state shown in table~\ref{tab:gauss_bs}\subref{tab:gauss_l0_bs}. Similarly, bound states in the second channel with $\lambda_{22}=20$ are shown in 
table~\ref{tab:gauss_bs}\subref{tab:gauss_l1_bs}.
\begin{table}[h!]
	\centering
	\subtable[First channel, $\lambda_{11}=7$]{
		\begin{tabular}{r|r|r}
			$i$ & $k_{i}$ (\textsc{Newton}) & $k_{i}$ (\textsc{Matslise}) \\
			\hline
			0 & \texttt{2.185562e+00i} & \texttt{2.185562e+00i}
		\end{tabular}
		\label{tab:gauss_l0_bs}
	}
	\hspace{0.5cm}
	\subtable[Second channel, $\lambda_{22}=20$]{
		\begin{tabular}{r|r|r}
			$i$ & $k_{i}$ (\textsc{Newton}) & $k_{i}$ (\textsc{Matslise}) \\
			\hline
			0 & \texttt{3.617543e+00i} & \texttt{3.617478e+00i} \\
			1 & \texttt{8.938842e-01i} & \texttt{8.938842e-01i}
		\end{tabular}
		\label{tab:gauss_l1_bs}
	}
	\caption{Bound states of uncoupled channels in the $s$- and $p$-wave example, obtained with (1) Newton iterations 
			on $\det(\mathbold{F}(k,\lambda))$ (\textsc{Newton}) and (2) with \textsc{Matslise}.}
	\label{tab:gauss_bs}
\end{table}
These results were calculated through Newton iterations on the function $\det(\mathbold{F}(k,\lambda))$. They compare, see 
table \ref{tab:gauss_bs}, favorably with results obtained with \textsc{Matslise}~\cite{Ledoux2005}. 

\begin{table}[h!]
	\centering
	\begin{tabular}{r|r|r|r}
		$i$ & $k_{i}$ (\textsc{Newton}) & $k_{i}$ (\textsc{Matscs}) & $k_{i}$ (\textsc{fin.\ diff.}) \\
		\hline
		0 & \texttt{3.623677e+00i} & \texttt{3.623551e+00i} & \texttt{3.623677e+00i} \\
		1 & \texttt{2.178012e+00i} & \texttt{2.178012e+00i} & \texttt{2.195645e+00i} \\
		2 & \texttt{9.035406e-01i} & \texttt{9.035275e-01i} & \texttt{9.039594e-01i}
	\end{tabular}
	\caption{Coupled channel bound states in the $s$- and $p$-wave example, with Newton iteration (\textsc{Newton}), 
			coupled channel extension of CPM\{$P$,$N$\} (\textsc{Matscs}) and finite differences (\textsc{fin.\ diff.})}
	\label{tab:gauss_l0_l1_c05_bs}
\end{table}
The coupled system on the other hand has bound states at positions deviating from those in table~\ref{tab:gauss_bs} as the coupling strength increases. In this case, with $\lambda_{12} = 0.5$, our results, see table~\ref{tab:gauss_l0_l1_c05_bs}, are in close agreement with 
those obtained through the recently extended CPM\{$P$,$N$\} method for coupled channel systems \cite{Ledoux2006,Ledoux2007} 
implemented in \textsc{Matscs}~\cite{Matscs}. Reference values obtained with eigenvalue calculations based on finite difference discretization, 
although not very accurate, are also provided.

It is now interesting to see how these bound states evolve in terms of variations in $\lambda_{22}$. Applying numerical continuation with 
starting values $k_{i}$ (\textsc{Newton}) taken from table~\ref{tab:gauss_l0_l1_c05_bs} gives results shown in figure~\ref{fig:gauss_l0_l1_c05_continuation}. 
The continuation was performed both in the direction of the origin and away to show as much of the resonant trajectories as possible.

Details of the underlying numerical procedures are as follows. The integration interval of the coupled channel Numerov method spans from 0 to 4.6 and 
contains 4096 points. The tolerance for Newton iterations used by AUTO is of the order $10^{-6}$ whereas the length of the continuation steps might vary adaptively between $10^{-2}$ and $10^{-4}$.

Each of the three continuation branches originates from a bound state and follows a path towards the resonant region (i.e.\ the lower half of the complex $k$-plane) 
as the potential strength $\lambda_{22}$ is decreased. Near the origin of the complex plane the second channel potential becomes too weak to hold the bound state 
and reaches a threshold value where the state transforms into a resonance. Looking at the continuation curves in figure~\ref{fig:gauss_l0_l1_c05_continuation} we 
can therefore distinguish three regions: the bound state regime on the positive imaginary axis, the transition into a resonance around the origin and the resonant 
regime below the real axis. Their properties are summarized as follows. (1) The bound states are situated on the positive imaginary axis and correspond to real, 
negative (rel.\ to potential asymptote) energies. In the neighborhood of the bound states of the first channel, in this case $k\approx 2.2i$, ``avoided crossings'' can 
be seen where a relatively big decrease in $\lambda_{22}$ has a minor influence on $k_{i}$. However, as soon as the lower state $k_{i-1}$ approaches, the 
higher one gets pushed away from this position and makes place for $k_{i-1}$. (2) As $\lambda_{22}$ passes certain threshold values, the potential becomes 
too weak to hold a bound state $k_{i}$ which transitions into a resonance. For $s$-waves these transitions happen on the negative imaginary axis not far from 
the origin, whereas for higher angular momenta this branching occurs exactly at 0. In the case of $s$/$p$-wave coupling, both effects influence each other 
depending on the coupling strength and bound states of the separate channels. The threshold values for $\lambda_{22}$ with the corresponding branch 
points for the three states in our example are summarized in table~\ref{tab:branchpoints}. (3) In the resonant regime ($\text{Im}(k)<0$) the states gain a real part which 
corresponds to a positive scattering energy. The negative imaginary part of these states is related to the time delay in the time-dependent picture. These 
branches are symmetric but only those in the fourth quadrant have physical significance. On the negative imaginary axis one also finds
so called virtual states.

\begin{table}[h!]
	\centering
	\begin{tabular}{r|r|r|r}
		$i$ & $\text{Re}(k_{i})$ & $\text{Im}(k_{i})$ & $\lambda_{22}^{t}$ \\
		\hline
		$1$ & \texttt{-1.227997e-12} & \texttt{-7.999927e-04} & \texttt{6.091215e+00} \\
		$2$ & \texttt{6.280509e-13} & \texttt{-1.589331e-02} & \texttt{1.742094e+01}
	\end{tabular}
	\caption{Branch points at which bound states transition into a resonance for the Gauss $s$/$p$-wave example. The real part of $k$ is negligible whereas the imaginary part indicates slight off-zero branching discussed in~\ref{sec:model2}.}
	\label{tab:branchpoints}
\end{table}

The proposed regularization \eqref{eq:det_multi_channel_regularized} plays a significant role near and at threshold values. Indeed, we have proven in appendix~\ref{app:regularization} the absence of 
additional, spurious zeros of $\det(\mathbold{F}(k,\lambda))$. All $\lambda_{22}$-dependent poles (which would otherwise coalesce with zeros in the origin at threshold values) are 
eliminated as well which makes continuation through the origin possible. We can not, however, guarantee smoothness of $\det(\mathbold{F}(k,\lambda))$ 
over the whole domain as demonstrated by the two stationary poles in figure~\ref{fig:gauss_l0_l1_c05_smatrix}, where we show a comparison 
of $\det(\mathbold{S}(k,\lambda))$ and $\det(\mathbold{F}(k,\lambda))$.

\begin{figure}[h!]
	\centering
	\subfigure[$\text{Im}(k)\times\text{Re}(k)$ projection]{
		\includegraphics[height=9cm]{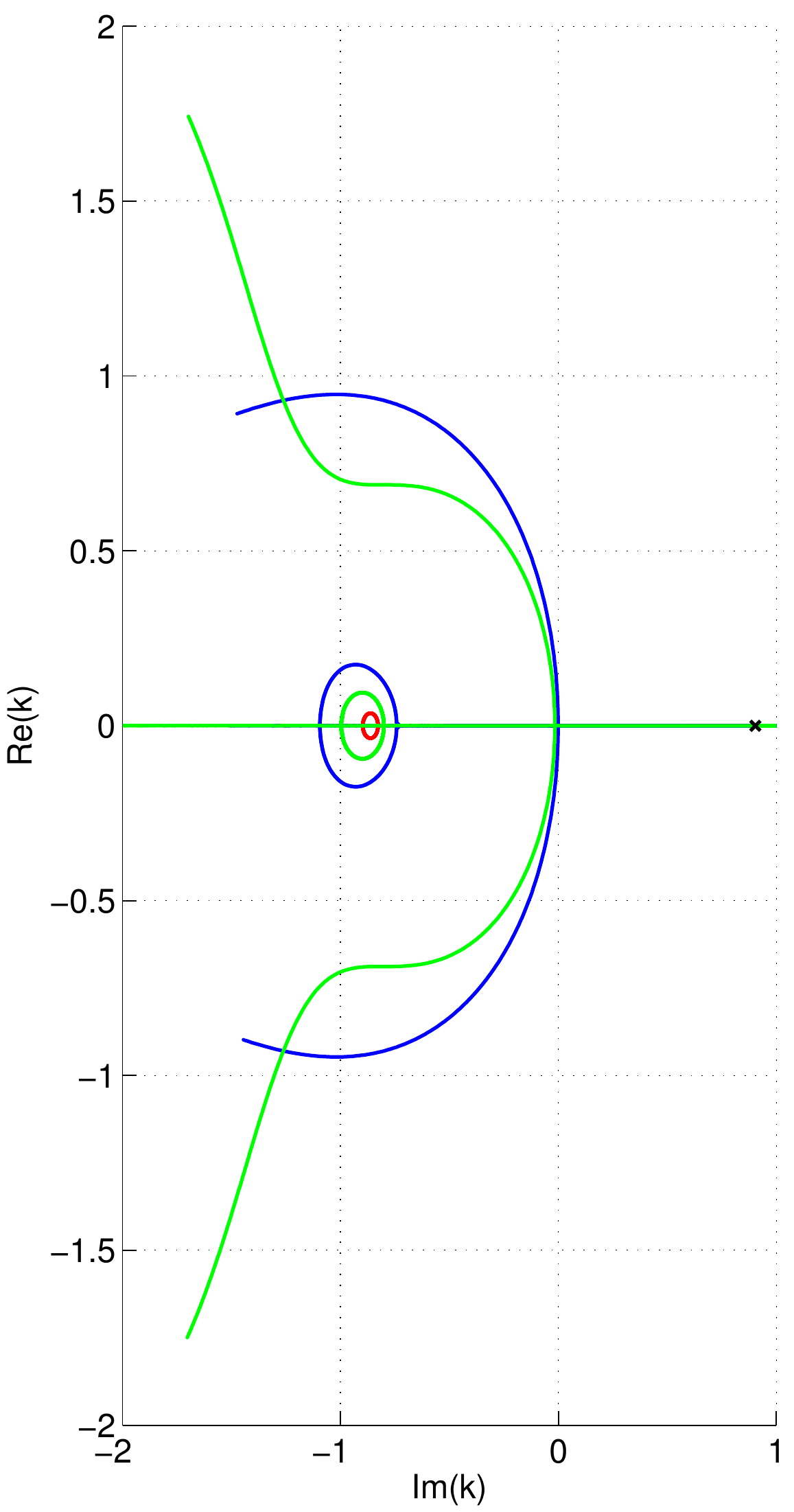}
		\label{fig:gauss_l0_l1_c05_top2}
	}
	\subfigure[$\text{Im}(k)\times\lambda_{22}$ projection. Bifurcation points are indicated with arrows.]{
		\includegraphics[height=9cm]{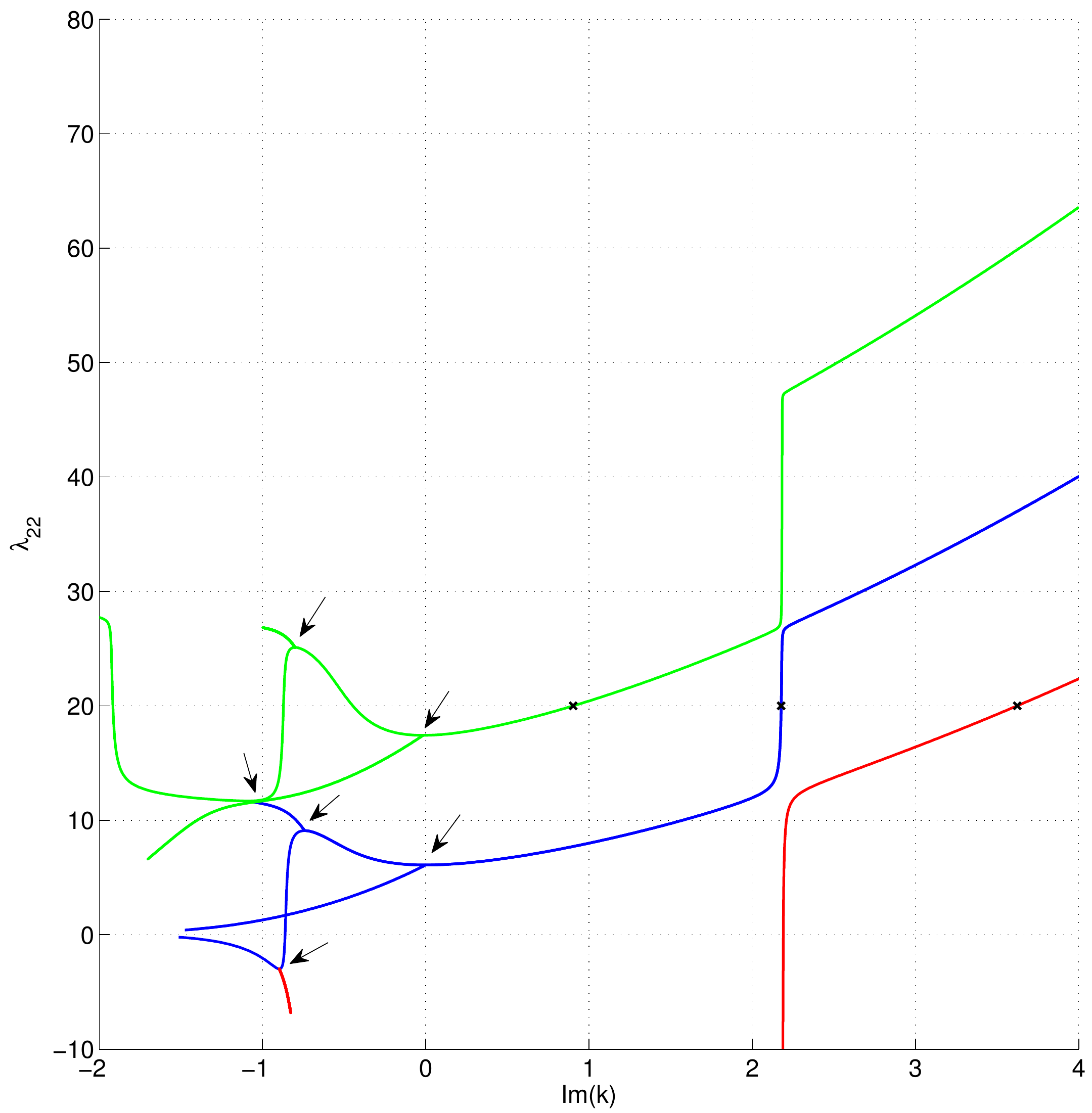}
		\label{fig:gauss_l0_l1_c05_side1}
	}
	\\
	\subfigure[Continuation curves in full space with projections along the sides.]{
		\includegraphics[height=11cm]{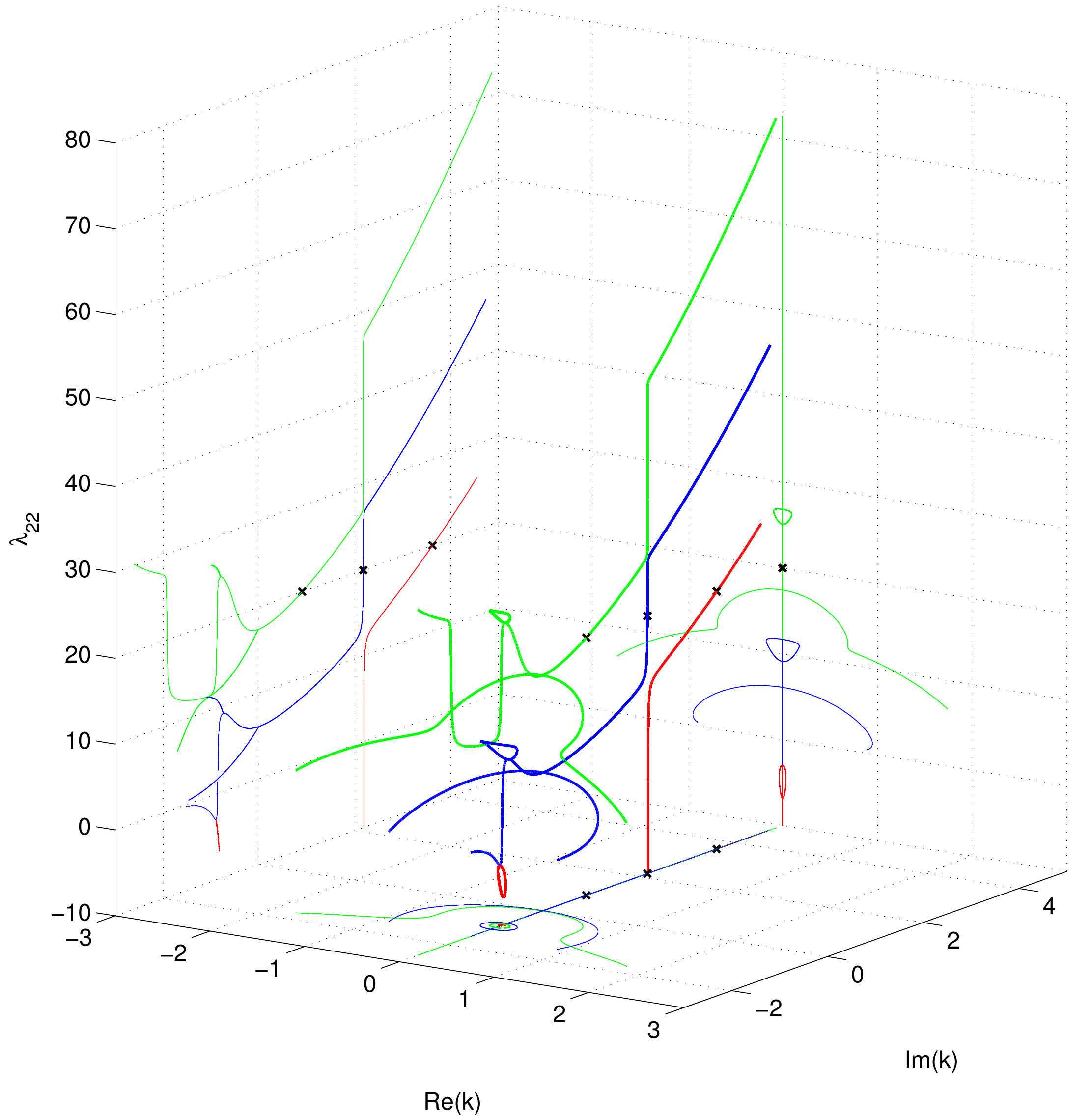}
		\label{fig:gauss_l0_l1_c05_persp}
	}
	\caption{Continuation curves of the Gauss $s$/$p$-wave coupling example. Branches of different colors originate from different starting points 
	(drawn as black crosses) in table~\ref{tab:gauss_l0_l1_c05_bs}: $i=0$ (red), $i=1$ (blue), $i=2$ (green)}
	\label{fig:gauss_l0_l1_c05_continuation}
\end{figure}

\begin{figure}[h!]
	\centering
	\begin{minipage}[t]{0.45\linewidth}
		\centering
		\subfigure{\includegraphics[height=3.2cm]{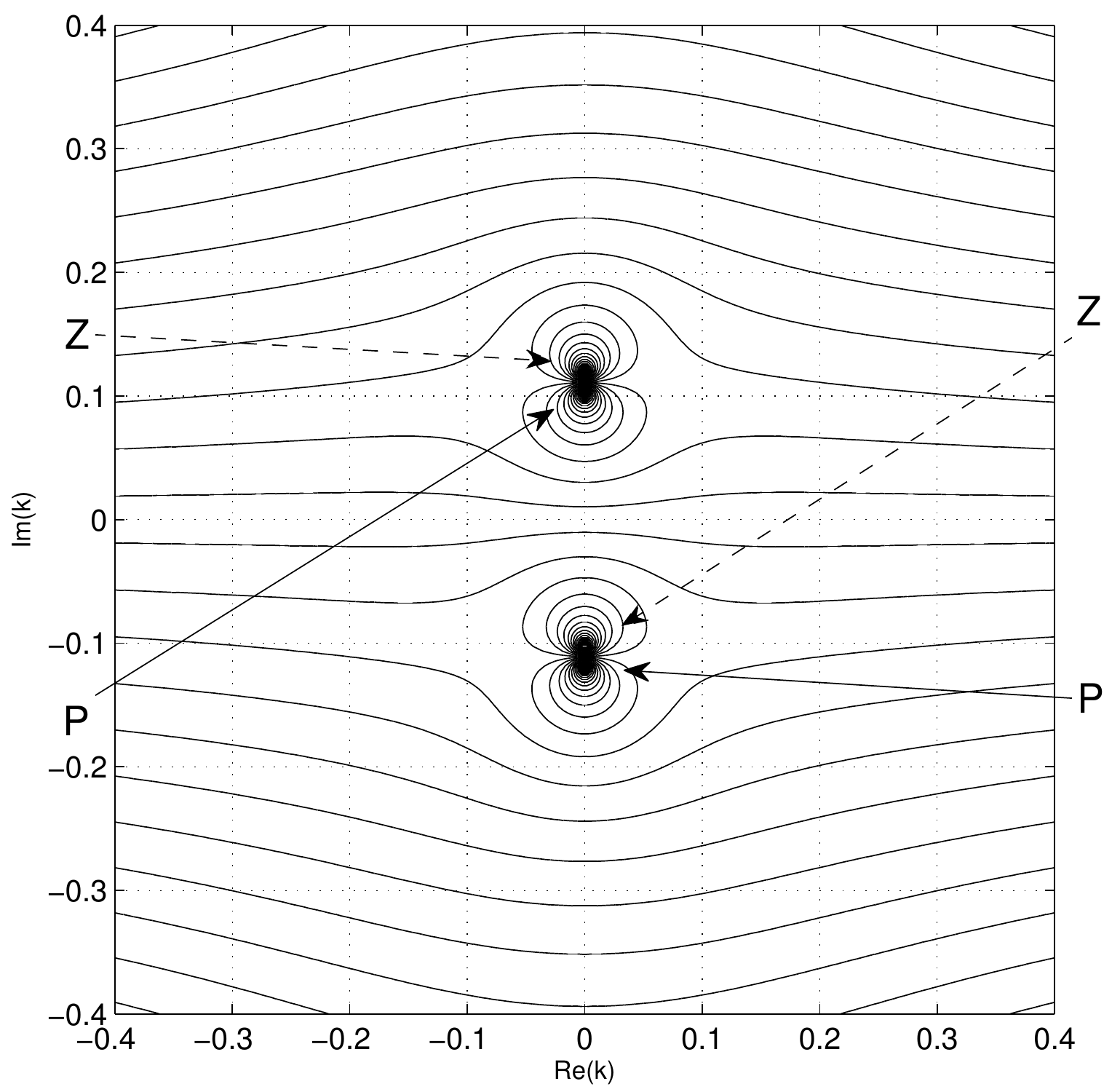}}
		\subfigure{\includegraphics[height=3.2cm]{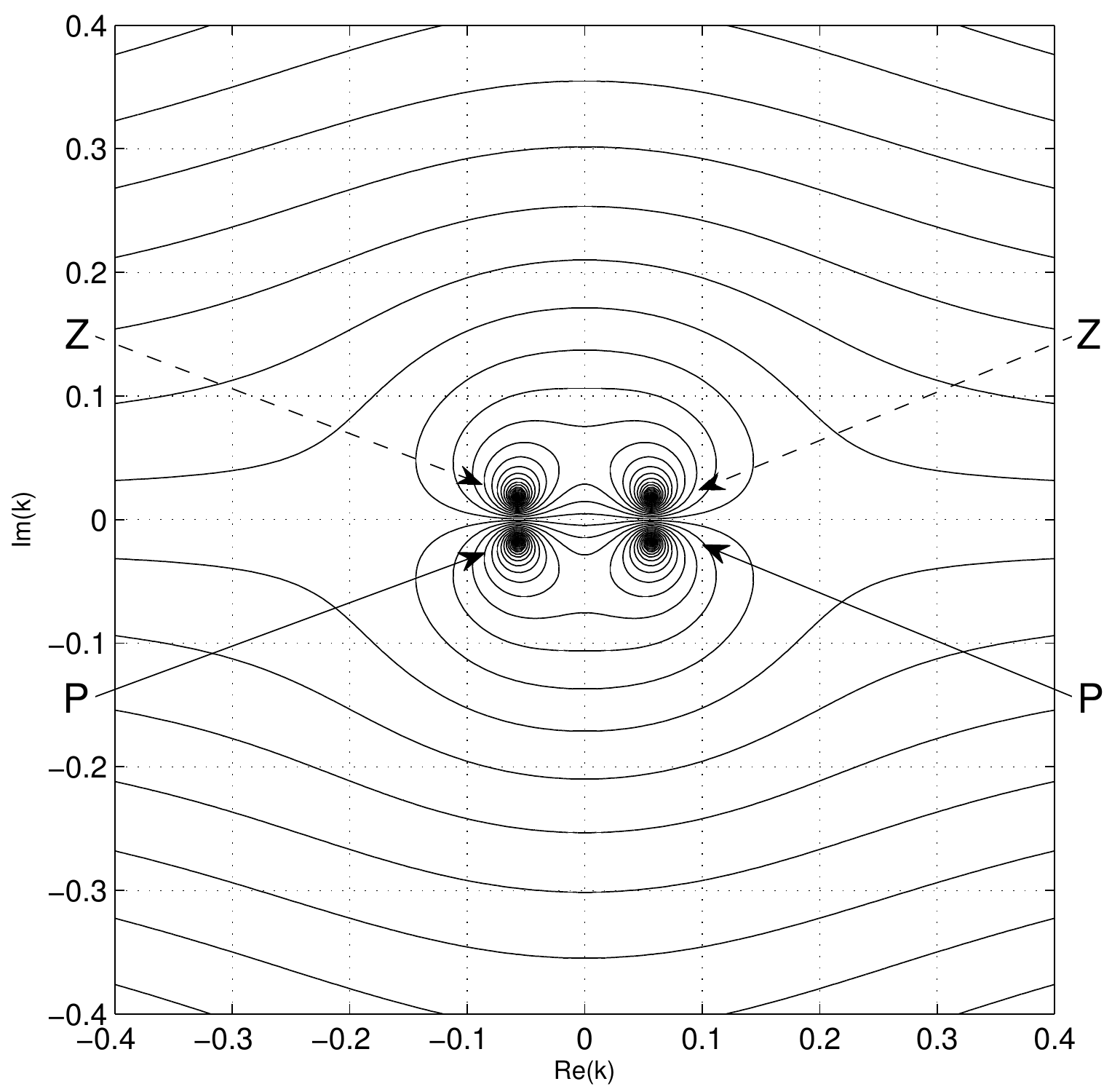}}\\
		\subfigure{\includegraphics[height=3.2cm]{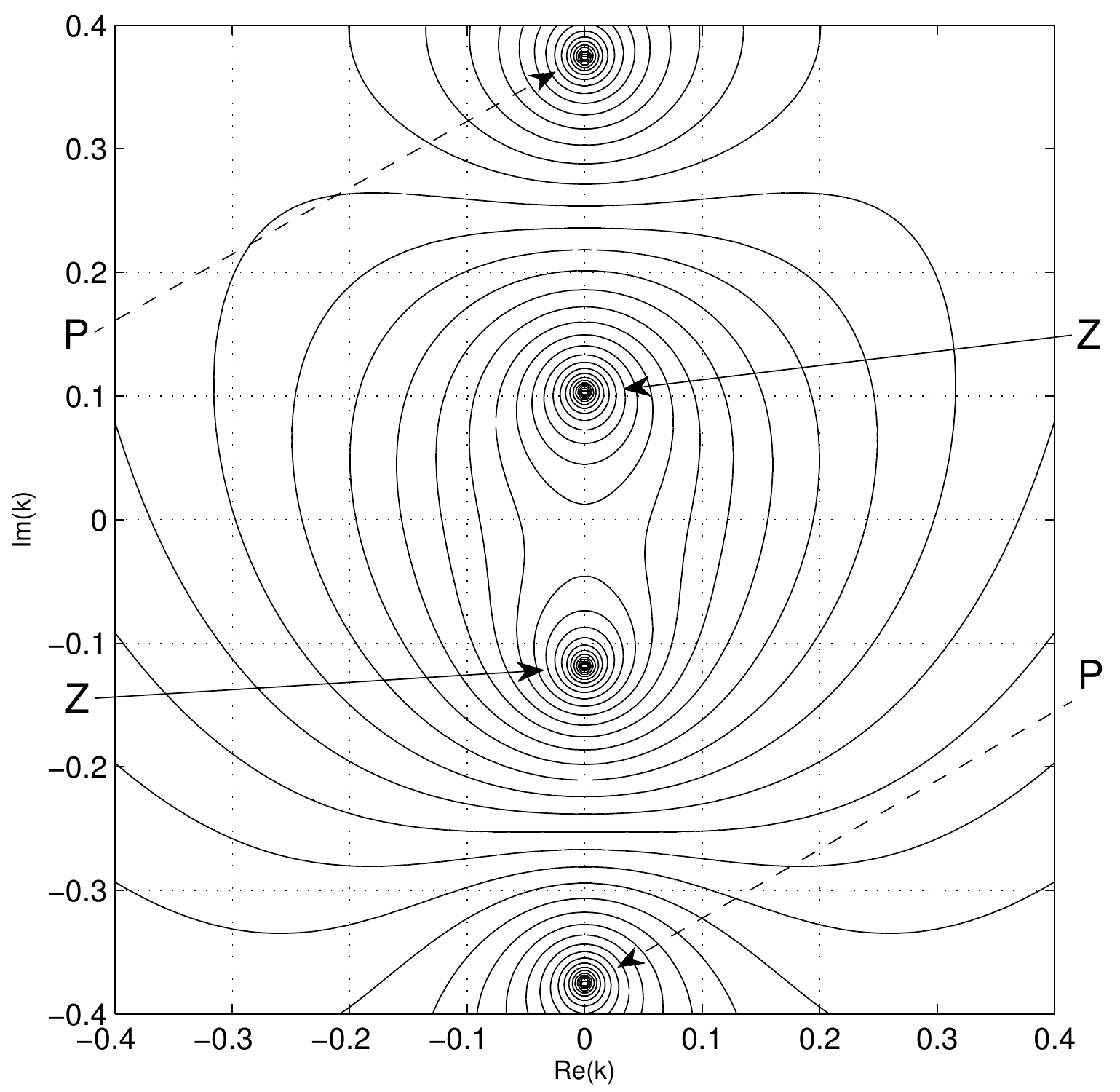}}
		\subfigure{\includegraphics[height=3.2cm]{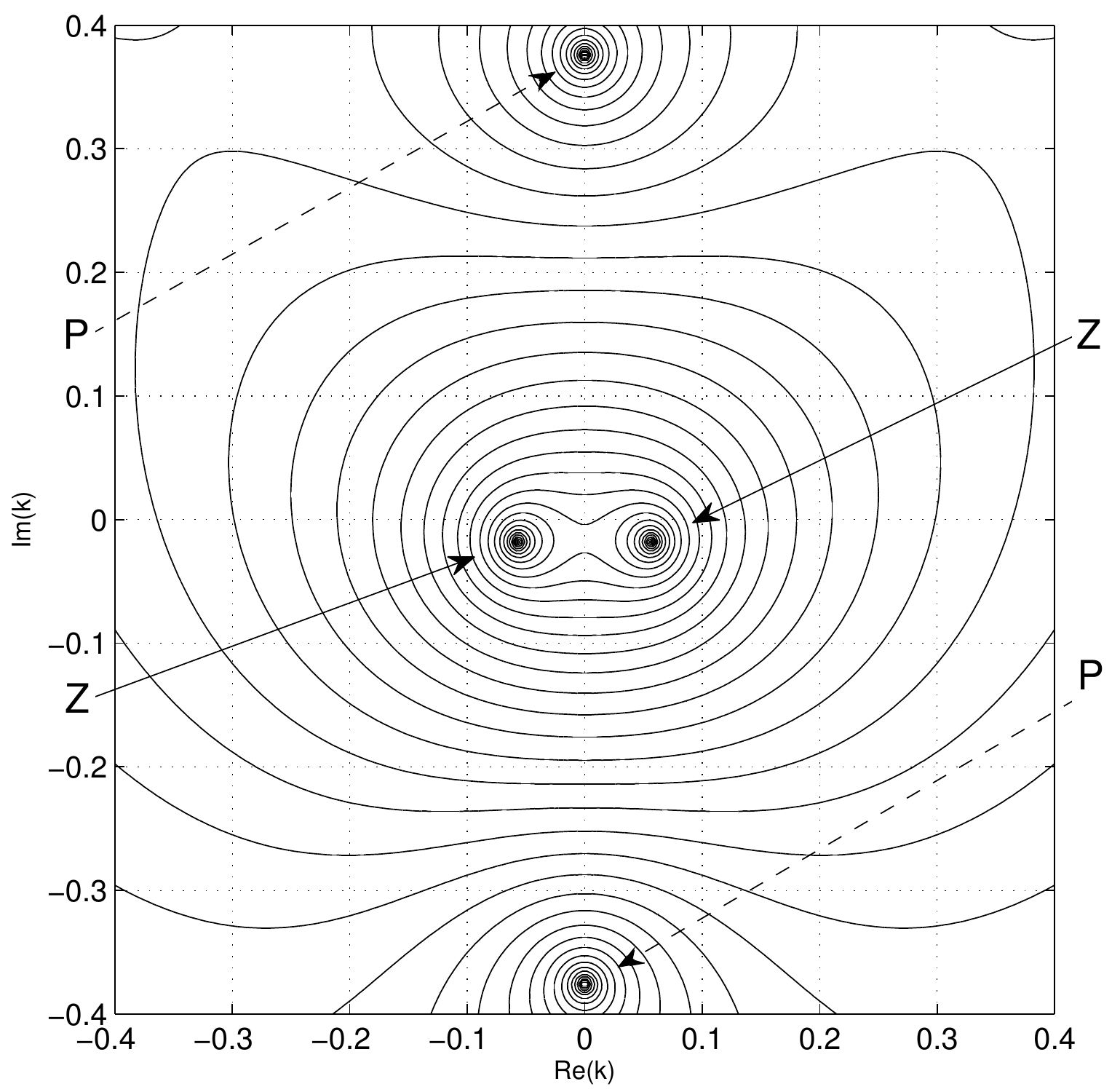}}
		\caption{Contour plots of $|\det(\mathbold{S})|$ (upper row) and $|\det(\mathbold{F})|$ (lower row) in the Gauss $s$/$p$-wave example (logarithmic scaling was applied to ``flatten'' the images). $\lambda_{12}=0.5$, $\lambda_{11}=7$ and $\lambda_{22}=17.5$ (left) and $17.4$ (right). As $\lambda_{22}$ is decreased, poles of the $S$-matrix move from the bound state and virtual state region through the origin of the complex $k$-plane into the resonant regime ($i=2$ state becomes resonant). Arrows indicate poles (P) and zeros (Z).}
		\label{fig:gauss_l0_l1_c05_smatrix}
	\end{minipage}
	\hspace{0.5cm}
	\begin{minipage}[t]{0.45\linewidth}
		\centering
		\subfigure{\includegraphics[height=3.2cm]{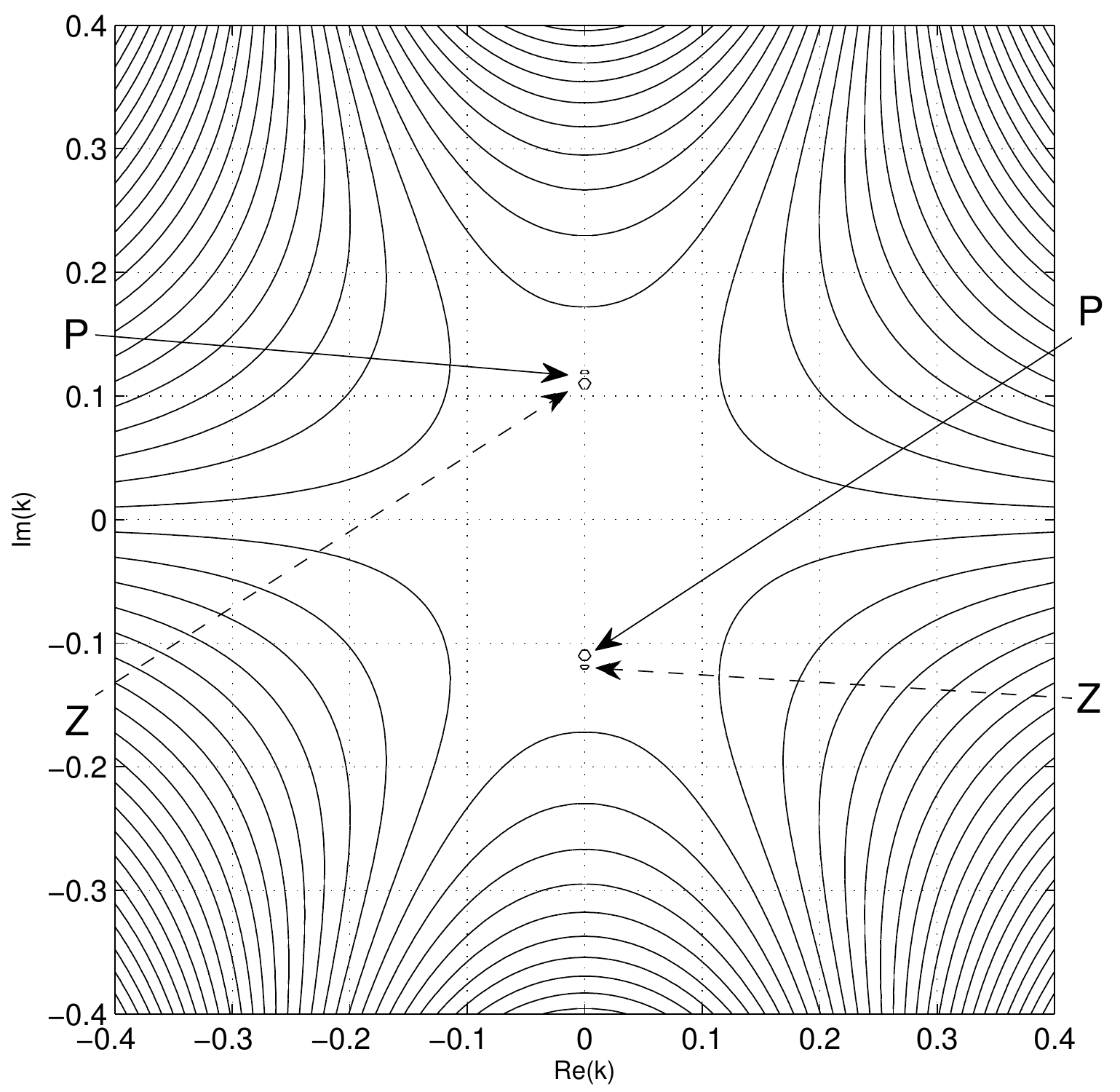}}
		\subfigure{\includegraphics[height=3.2cm]{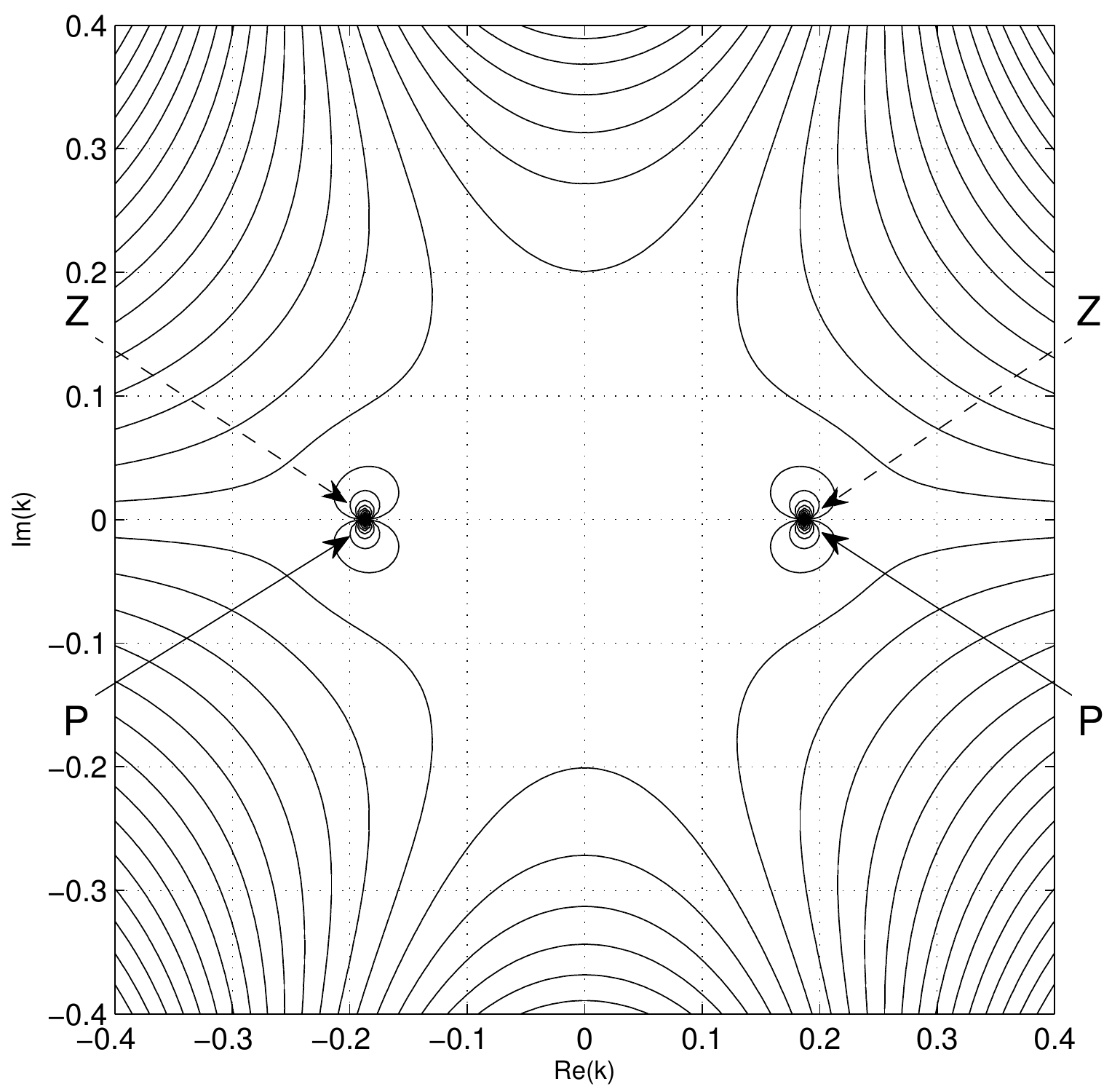}}\\
		\subfigure{\includegraphics[height=3.2cm]{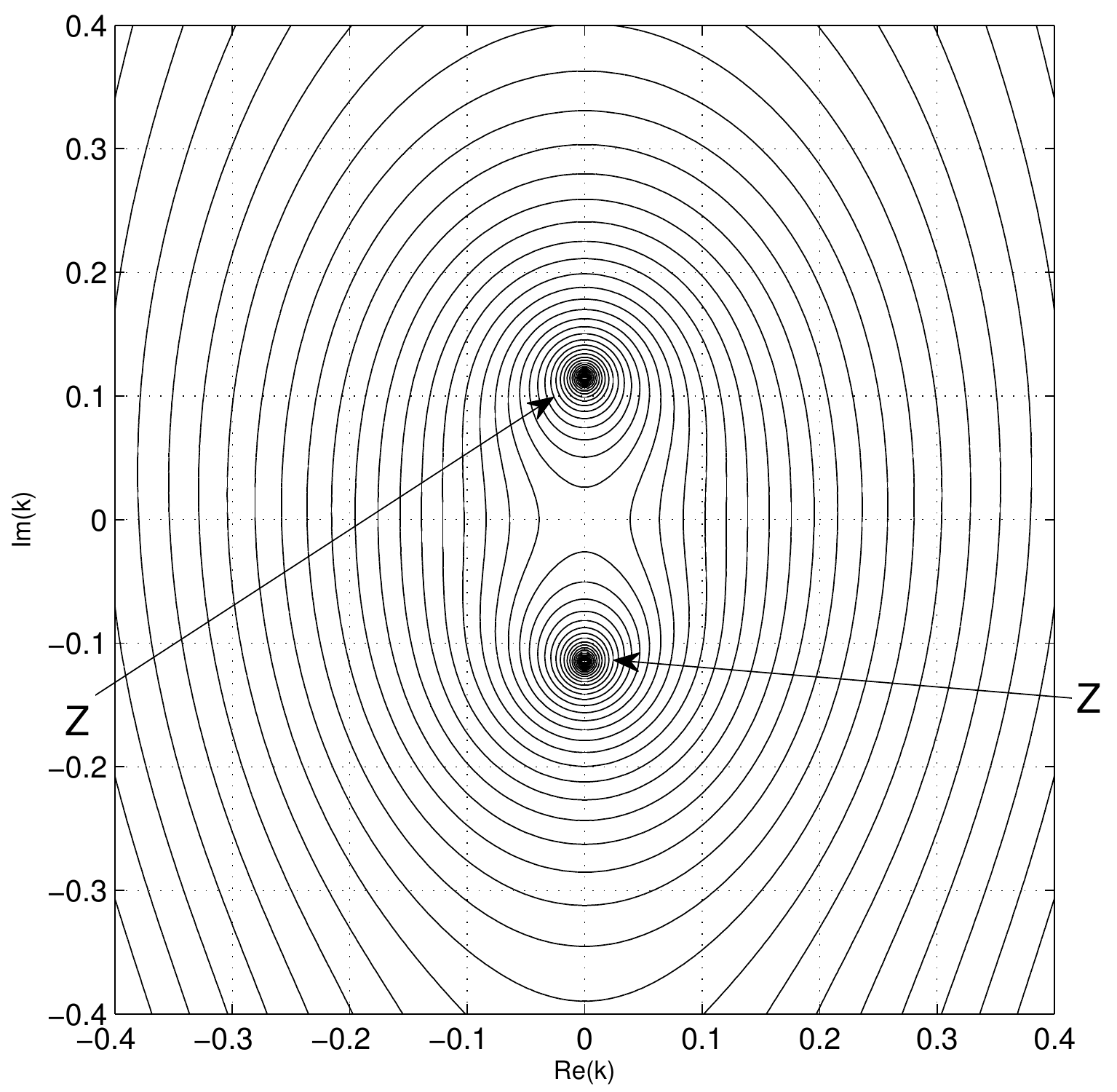}}
		\subfigure{\includegraphics[height=3.2cm]{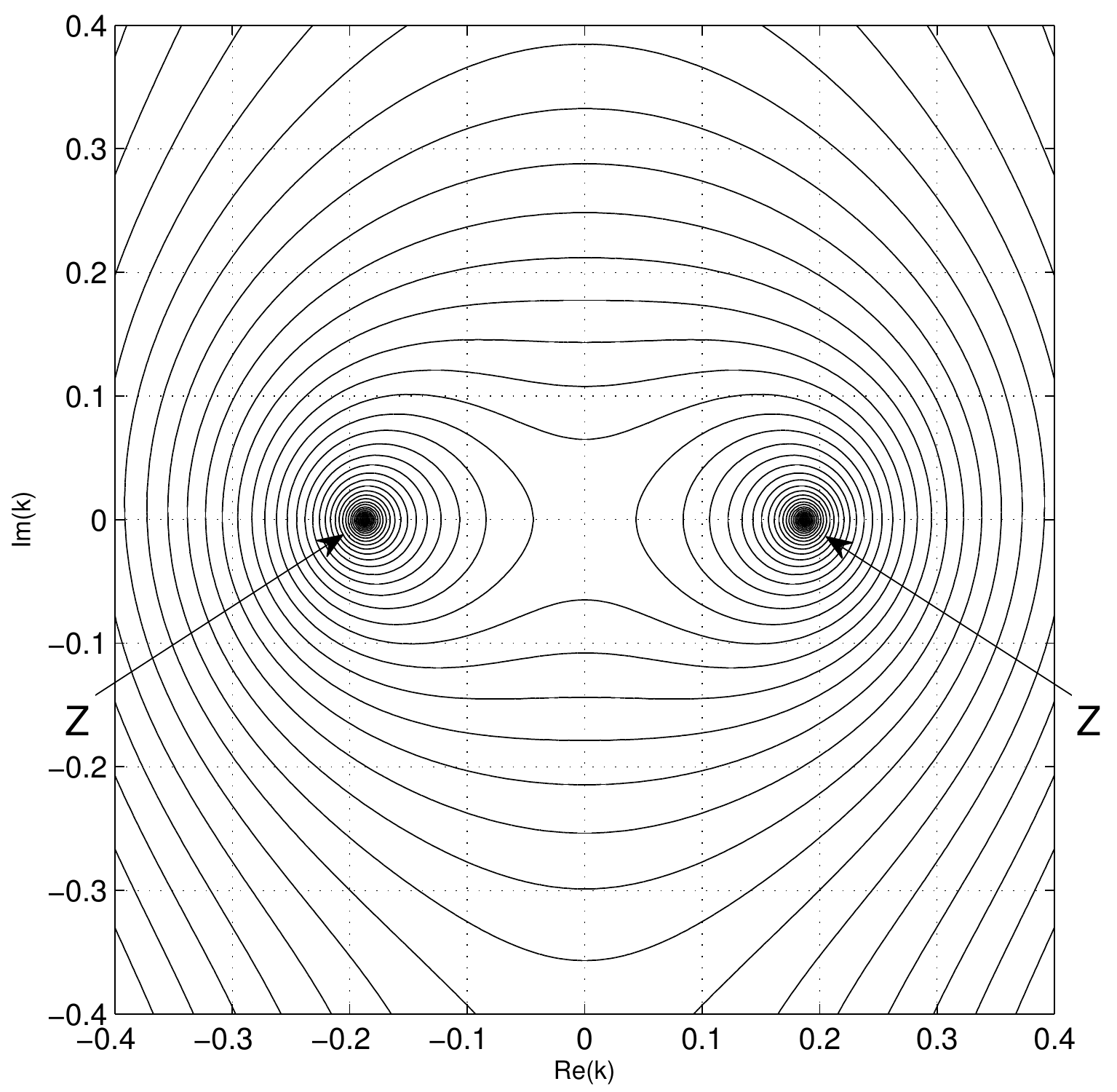}}
		\caption{Contour plots of $|\det(\mathbold{S})|$ (upper row) and $|\det(\mathbold{F})|$ (lower row) in the Gauss $p$/$d$-wave example (logarithmic scaling was applied to ``flatten'' the images). $\lambda_{12}=0.3$, $\lambda_{11}=10$ and $\lambda_{22}=13.5$ (left) and $13.4$ (right). As $\lambda_{22}$ is decreased, poles of the $S$-matrix move from the bound state and virtual state region through the origin of the complex plane into the resonant regime ($i=1$ state becomes resonant). Arrows indicate poles (P) and zeros (Z).}
		\label{fig:gauss_l1_l2_c03_smatrix}
	\end{minipage}
\end{figure}

\subsection{Gauss potential with $p$-wave and $d$-wave coupling}
\label{sec:model3}
The second model problem also uses Gauss potentials in both channels as well as the coupling and $\mu=1$. However, angular momenta for channel 1 and 2 were chosen to be $l_{1}=1$ and $l_{2}=2$. This choice is motivated by the previously mentioned distinction between $s$-waves and waves with higher angular momenta in \ref{sec:model2}: In the $l=0$ case the bifurcation does not take place in the origin of the $k$-plane but on the imaginary axis below the origin, whereas for $l>0$ it does branch off in the origin \cite{Taylor2006}. By taking both angular momenta higher than zero we enforce branching of the coupled channel system in the origin and therefore an essentially different geometry for testing the regularized formulation than in the previous example.

Setting $\lambda_{11}=10$, $\lambda_{22}=30$ and $\lambda_{12}=0.3$ gives a system with 3 bound states. The approximations, analogous to the first example, are given in table~\ref{tab:gauss_l1_l2_c03_bs}.
\begin{table}[h!]
	\centering
	\begin{tabular}{r|r|r|r}
		$i$ & $k_{i}$ (\textsc{Newton}) & $k_{i}$ (\textsc{Matscs}) & $k_{i}$ (\textsc{fin.\ diff.}) \\
		\hline
		0 & \texttt{3.796532e+00i} & \texttt{3.796472e+00i} & \texttt{3.796532e+00i} \\
		1 & \texttt{1.600083e+00i} & \texttt{1.600083e+00i} & \texttt{1.600127e+00i} \\
		2 & \texttt{6.599123e-01i} & \texttt{6.599112e-01i} & \texttt{6.603354e-01i}
	\end{tabular}
	\caption{Bound states of the coupled channel $p$- and $d$-wave example. We compare values obtained by Newton iterations (\textsc{Newton}) on $\det(\mathbold{F}(k,\lambda))$, (\textsc{Matscs}) and finite differences (\textsc{fin. diff.})}
	\label{tab:gauss_l1_l2_c03_bs}
\end{table}

Again, these values were used as starting points for numerical continuation. The resulting curves are shown in figure~\ref{fig:gauss_l1_l2_c03_continuation}. The details of the coupled channel renormalized Numerov solver and AUTO continuation parameters were set to the same values as in the first example.

Also in this case the proposed regularization provides a viable way of removing poles in the origin at thresholds, hereby allowing continuation of the relevant zeros through the origin. This is shown in figure~\ref{fig:gauss_l1_l2_c03_smatrix}. This figure also shows the sharpness of nearby poles and zeros of $\det(\mathbold{S}(k,\lambda))$ compared to much smoother behavior of $\det(\mathbold{F}(k,\lambda))$ which positively influences the convergence of Newton iterations in the continuation process.

\begin{figure}[h!]
	\centering
	\subfigure[$\text{Im}(k)\times\text{Re}(k)$ projection]{
		\includegraphics[height=9cm]{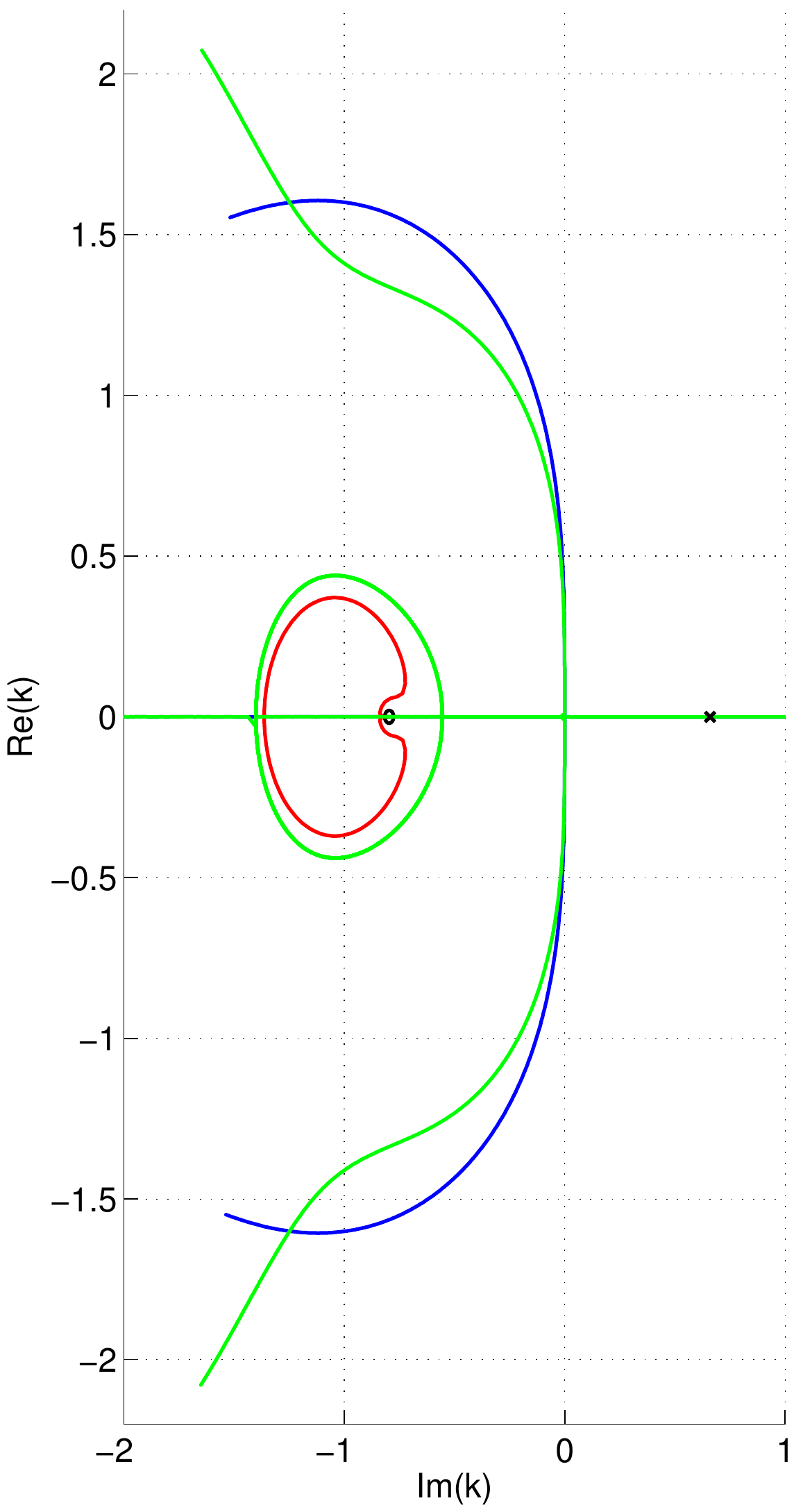}
		\label{fig:gauss_l1_l2_c03_top2}
	}
	\subfigure[$\text{Im}(k)\times\lambda_{22}$ projection. Bifurcation points are indicated with arrows.]{
		\includegraphics[height=9cm]{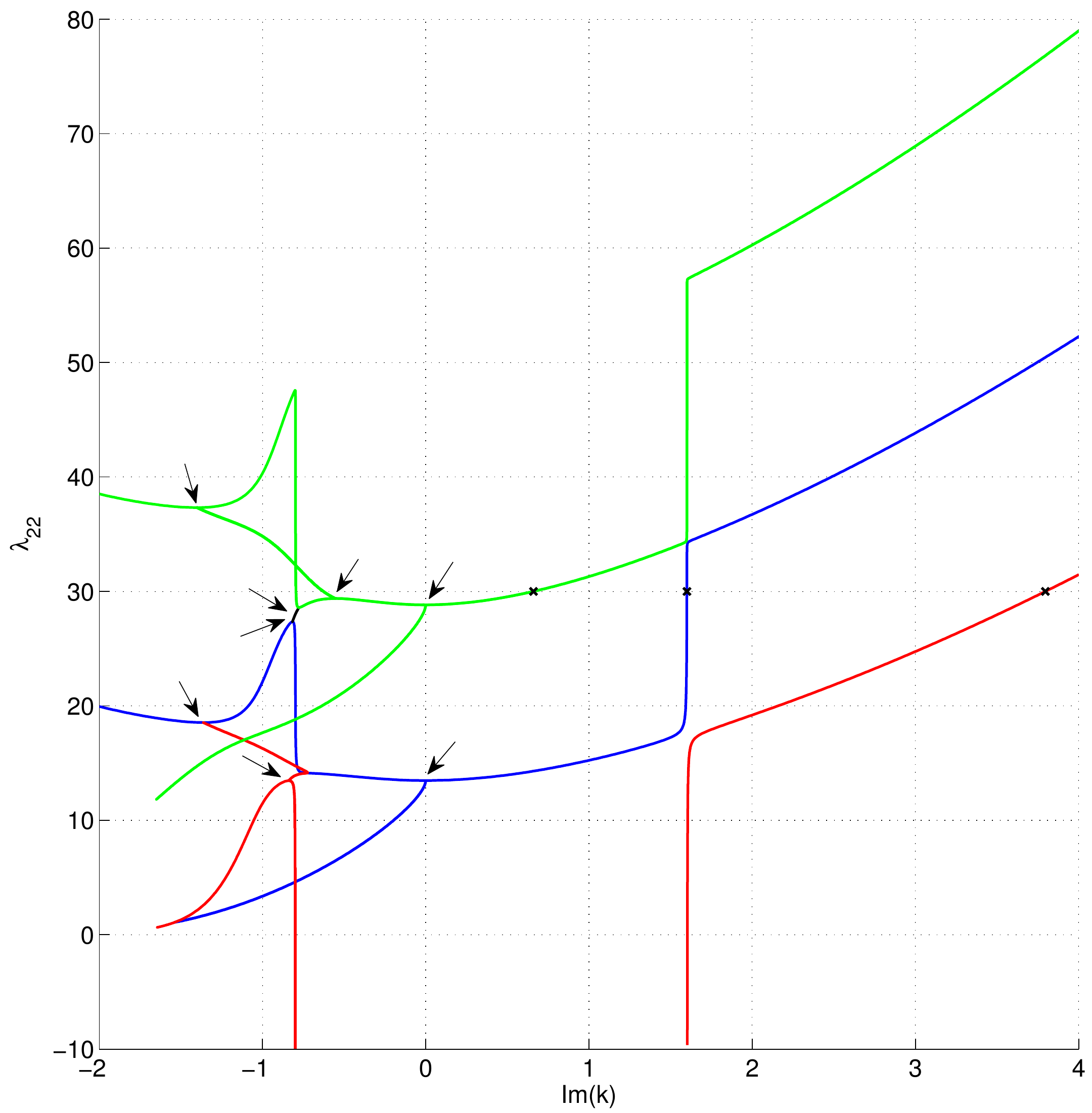}
		\label{fig:gauss_l1_l2_c03_side1}
	}
	\\
	\subfigure[Continuation curves in full space with projections along the sides.]{
		\includegraphics[height=11cm]{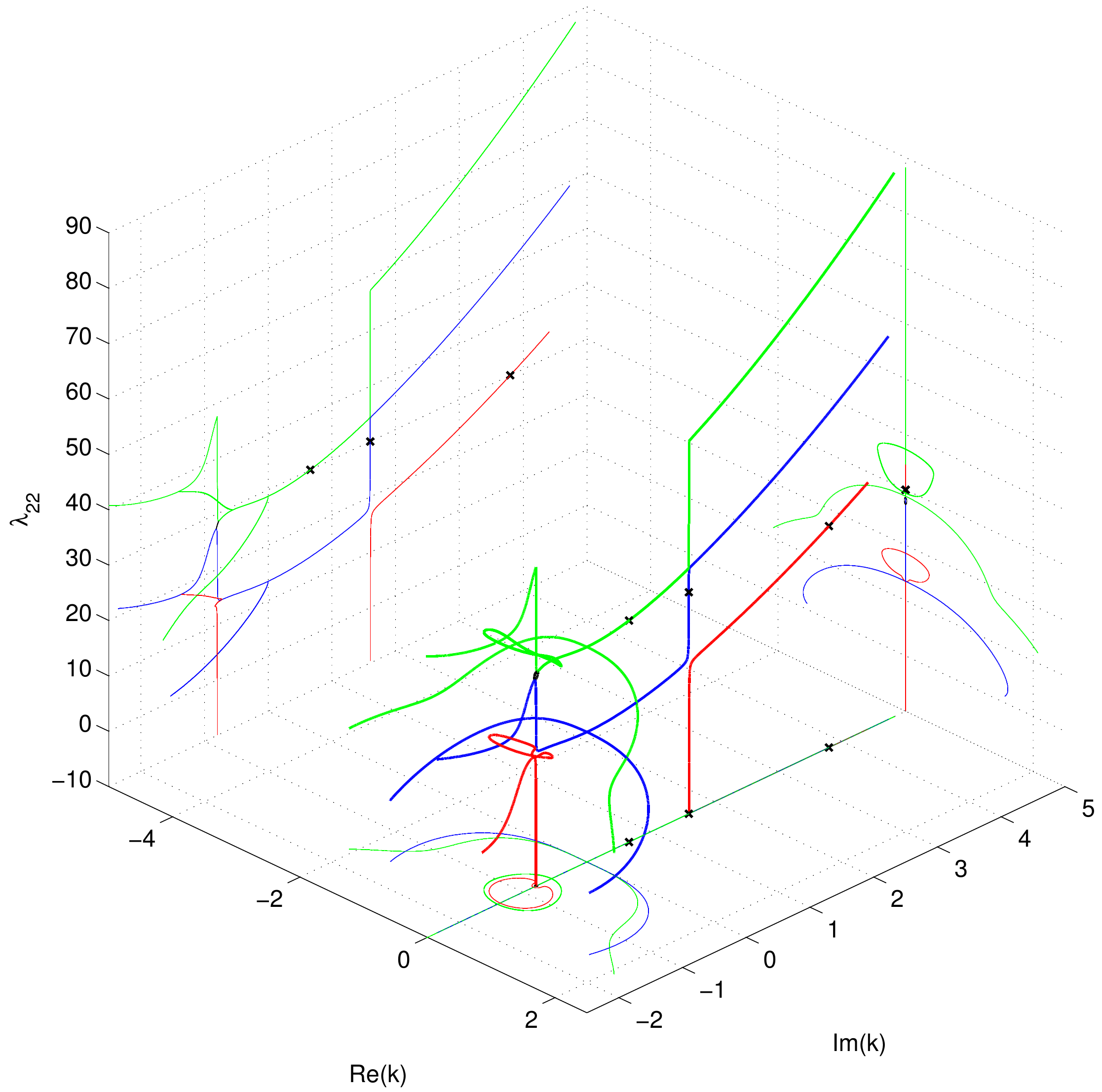}
		\label{fig:gauss_l1_l2_c03_persp}
	}
	\caption{Continuation curves of the Gauss $p$/$d$-wave coupling example. Branches of different colors originate from different starting points 
	(drawn as black crosses) in table~\ref{tab:gauss_l1_l2_c03_bs}: $i=0$ (red), $i=1$ (blue), $i=2$ (green)}
	\label{fig:gauss_l1_l2_c03_continuation}
\end{figure}

%
%
\section{Discussion and Conclusions}
In this article we have applied numerical continuation to track resonant states
in a coupled channel Sch\"odinger equation. In our approach, we track the zeros of a function that can be extracted from the numerical solutions 
of the coupled channel equation. This function, related to the $S$-matrix by a transformation, has a zero wherever the $S$-matrix has a pole and is amenable to
numerical continuation.
The analysis has also revealed that the function can have accidental poles which may slow down the convergence of the Newton iterations. However, for our model 
problems, we have not seen any difficulties caused by this. 

We have tested this approach and solved a range of two channel model problems with robust results. The method automatically tracked all the resonant 
states and detected the bifurcation points. It also successfully traversed the threshold which indicates a successful regularization.

Our motivation to develop this method originates from chemical reactions that are modeled by the Born-Oppenheimer approximation. It is based on slow-fast separation 
between the motion of nuclei and electrons that is justified because of their mass difference. In this approximation the modeling is a two step process. In the first step the 
scattering of the electrons in the field of the nuclei is solved with the positions of the nuclei as parameter in this scattering problem. It is especially important to identify the 
resonances in the scattering problem. In the second step the molecular dynamics of the nuclei, that move on the potential surface given by the resonances determined 
in the first step, is solved.

This paper is focused on the first step where the resonances need to be identified in a family of scattering problems each with a different parameter, which is the position 
of the nuclei in these problems.

In our future work we aim to create a robust and reliable tool that scales to realistic systems.

%
%

\section*{Acknowledgments}
We gratefully acknowledge the support from FWO-Vlaanderen through the project number G.0120.08.

\appendix
\section{Spherical functions}
\label{app:spherical}
Various solutions of the free radial wave equation are intensely used throughout this paper as they can describe the asymptotic behavior of the wave functions that solve the (non-free) Schr\"odinger equation. In an attempt to avoid ambiguities we summarize the used conventions in this appendix. For a far more thorough treatment we suggest \cite{Abramowitz1972} and \cite{Taylor2006}, where similar conventions are used.

Spherical Riccati-Bessel functions are defined as solutions of the free radial Schr\"odinger equation
\begin{equation}
	\label{eq:radial_schrodinger_free}
  	\left(\frac{d^{2}}{dr^{2}} - \frac{l(l+1)}{r^{2}} + k^{2} \right) u(r) = 0, \qquad k=\sqrt{2\mu E},\quad l\in\mathbb{N}.
\end{equation}
The different, commonly used particular solutions of this equation are: ($x=kr$)
\begin{itemize}
	\item $\hat{j}_{l}(x) \ = \ x\,j_{l}(x) \ = \ \sqrt{x\frac{\pi}{2}}J_{l+1/2}(x)$, regular solutions, spherical Riccati-Bessel functions. E.g.: $\hat{j}_{0}(x) = \sin(x),\ \hat{j}_{1}(x) = \frac{1}{x}\sin(x)-\cos(x),\ \ldots$
	\item $\hat{n}_{l}(x) \ = \ x\,n_{l}(x) \ = \ (-1)^{l+1}\sqrt{x\frac{\pi}{2}}J_{-l-1/2}(x)$, irregular, spherical Riccati-Neumann functions. E.g.: $\hat{n}_{0}(x) = -\cos(x),\ \hat{n}_{1}(x) = -\frac{1}{x}\cos(x)-\sin(x),\ \ldots$
	\item $\hat{h}_{l}^{\pm}(x) \ = \ -\hat{n}_{l}(x) \pm i\hat{j}_{l}(x)$ Incoming ($-$) and outgoing ($+$) spherical Riccati-Hankel functions. E.g.: $\hat{h}_{0}^{\pm}(x) = e^{\pm ix},\ \hat{h}_{1}^{\pm}(x) = (\frac{1}{x} \mp i)e^{\pm ix},\ \ldots$
\end{itemize}

\section{Regularization lemma}
\label{app:regularization}
This appendix technically describes the proof that the application of the regularizing procedure in section~\ref{sec:Regularization} does not introduce any zeros of the continuation function $\det(\mathbold{F}(k,\lambda))$ that can not be attributed to bound or resonant states. The proof involves the use of some well-known definitions and results from quantum mechanical scattering and gives a thorough explanation of the procedure yet is not necessary for the understanding and reproduction of the main results of our contribution. The conventions and definitions used are mostly those found in \cite{Taylor2006} and, on occasion in \cite{Newton1982,Newton1960}. Section \ref{appsec:jost} summarizes briefly a necessary minimal subset of definitions for the proof whereas section \ref{appsec:boundsproof} proves why the regularizing procedure does not introduce false resonances or bound states.

For clarity we omit the potential parameter $\lambda$ in our notations as it does not play a role in the derivations; unless it is really necessary or clarifying. Following~\cite{Taylor2006} we do however, indicate the explicit dependence of the wave functions on $k$, e.g.\ $\mathbold{\Psi}(k,r)$.

\subsection{The Jost matrix and the $S$-matrix}
\label{appsec:jost}
The ``regular solution'' $\mathbold{\Phi}(k,r)$ of equation~\eqref{eq:matrixeq} is defined to behave exactly as $\hat{\mathbold{j}}(kr)$ at $r=0$. This fixes both $\mathbold{\Phi}$ and its normalization uniquely. At infinity the regular solution must tend to a linear combination of two other linearly independent solutions such as the free solutions $\hat{\mathbold{h}}^{\pm}$, hence allowing us to write
\begin{equation}
	\mathbold{\Phi}(k,r) \xrightarrow{r\to\infty} \frac{i}{2}\Big( \hat{\mathbold{h}}^{-}(kr)\mathbold{f}(k) \, - \, \hat{\mathbold{h}}^{+}(kr){\mathbold{f}(k)}^{*} \Big),
\end{equation}
which defines the ``Jost matrix'' $\mathbold{f}(k)$ as a diagonal matrix that bundles so called ``Jost functions'' $f_{l}(k)$
\begin{equation}
	\mathbold{f}(k) = \begin{pmatrix}
		f_{0}(k) \\
		& \ddots \\
		&& f_{l_{\text{max}}}(k)
	\end{pmatrix}.
\end{equation}
(it is known that ${f_{l}(k)}^{*} = f_{l}(-k)$ and hence ${\mathbold{f}(k)}^{*} = \mathbold{f}(-k)$) Since we have already written the asymptotic form of the ``physical solution'' $\mathbold{\Psi}(k,r)$ (which is calculated by the numerical solver) in terms of the $S$-matrix as in
\begin{equation}
	\mathbold{\Psi}(k,r) \xrightarrow{r\to\infty} \frac{i}{2}\left( \hat{\mathbold{h}}^{-}(kr) - \hat{\mathbold{h}}^{+}(kr)\mathbold{S}(k) \right),
\end{equation}
it follows that the $S$-matrix can be expressed as the ratio of the Jost matrices
\begin{equation}
	\mathbold{S}(k) = \mathbold{f}(-k){\mathbold{f}(k)}^{-1}
\end{equation}
and that $\mathbold{\Phi}$ and $\mathbold{\Psi}$ are related through the well-known expression\footnote{Please note the different approaches studied in various reference works on the subject. Some authors such as \cite{Taylor2006} prefer the simpler approach of expression \eqref{eq:physical_vs_regular}, whereas others like \cite{Newton1982} explicitly write certain factors without incorporating them, for instance, in the Jost function. In the single channel variant of \eqref{eq:physical_vs_regular} for instance this would result in $\psi_{l}(k,r) = \frac{k^{l+1}}{(2l+1)!!}\frac{\phi_{l}(k,r)}{f_{l}(k)}$. Unless explicitly stated otherwise, we choose to follow the former convention.}
\begin{equation}
	\label{eq:physical_vs_regular}
	\mathbold{\Psi}(k,r) = \mathbold{\Phi}(k,r){\mathbold{f}(k)}^{-1}.
\end{equation}
The previous equations illustrate a widely accepted result of bound and resonant states being represented, equivalently, by poles of the $S$-matrix and zeros of the Jost function. As we can not extract the Jost matrix directly from the numerical solution of \eqref{eq:matrixeq} but only a function that is proportional the Jost matrix, we have to use the poles of the $S$-matrix as the characterization of bound and resonant states.

The physical solution $\mathbold{\Psi}$ can also be written as the solution of the coupled channel partial wave version of the Lippmann-Schwinger equation
\begin{equation}
	\label{eq:lippmann-schwinger}
	\mathbold{\Psi}(r) = \hat{\mathbold{j}}(kr) + \int_{0}^{\infty}dr'\, \mathbold{G}(r,r';k)\mathbold{U}(r')\mathbold{\Psi(r')},
\end{equation}
where $\mathbold{G}(r,r';k)$ denotes the diagonal matrix of one-channel Green's functions $G_{l_{i}}^{0}(r,r';k) = \frac{-1}{k}\hat{j}_{l}(kr_{<})\hat{h}_{l}^{+}(kr_{>})$, with $r_{<} = \min(r,r')$ and $r_{>} = \max(r,r')$ and where $\mathbold{U}=2\mu\mathbold{V}$.
If we now combine \eqref{eq:lippmann-schwinger} with the asymptotic behavior of \eqref{eq:matrixboundary2} we find
\begin{equation}
	\mathbold{S}(k) - \mathbold{I} = \frac{-2i}{k} \int_{0}^{\infty} dr'\, \hat{\mathbold{j}}(kr') \mathbold{U}(r') \mathbold{\Psi}(r').
\end{equation}
Inside the integral we replace $\mathbold{\Psi}$ with the regular solution $\mathbold{\Phi}$ through the identity~\eqref{eq:physical_vs_regular} which leads to the equation
\begin{equation}
	\label{eq:145907062010}
	\mathbold{S}(k) - \mathbold{I} = \frac{-2i}{k}\left( \int_{0}^{\infty} dr'\, \hat{\mathbold{j}}(kr') \mathbold{U}(r') \mathbold{\Phi}(r') \right){\mathbold{f}(k)}^{-1}.
\end{equation}

\subsection{Bounds on the regularization}
\label{appsec:boundsproof}
Equation~\eqref{eq:145907062010} states that the $(l_{\text{max}}+1) \times (l_{\text{max}}+1)$ matrix $\mathbold{S}(k)-\mathbold{I}$ is proportional to a matrix of integrals multiplied by the inverse of the Jost matrix. To guarantee that the regularizing procedure from section~\ref{sec:Regularization} defined by eq.~\eqref{eq:det_multi_channel_regularized} as
\begin{equation}
	\label{appeq:det_multi_channel_regularized}
	\det(\mathbold{F}(k,\lambda)) = \left( \prod_{l=0}^{l_{\text{max}}} k^{2l+1} \right) \bigg/ \det(\mathbold{S}(k,\lambda)-\mathbold{I}),
\end{equation}
does not introduce false bound states or resonances we must prove that zeros of the Jost matrix $\mathbold{f}(k)$ are the only contributors to the poles of expression~\eqref{eq:145907062010} and as such, to the zeros of $\det(\mathbold{F}(k,\lambda))$. This requires the matrix of integrals in equation~\eqref{eq:145907062010} to be bounded which is guaranteed by the following lemma.

An important remark has to be made ragarding the regularity constraints of the potential matrix $\mathbold{V}$. The proof of lemma~\ref{lem:integralbound} is a generalization of the one channel case where coupling terms in the potential do not occur. There, the boundary conditions that define the regular solution $\varphi_{l}$ and the iterative procedure of solving the one channel analogue of integral equation~\eqref{eq:lippmann-schwinger} for $\varphi_{l}$ indeed would give rise to proper bounds. Contrary to the one channel case however, the off-diagonal terms in $\mathbold{V}$ that couple channels with different angular momenta result in serious singularities in the origin $r=0$, hereby destroing the possibility for a bound without very restrictive assumptions on the potentials that would cancel out the singular behavior in the origin~\cite{Newton1982,Newton1960}. These restrictions can be lifted to the much more liberal requirement for the off-diagonal terms $\int_{0}^{\infty}dr\,r|V_{i\neq j}(r)| < \infty$ by rewriting the integral equation~\eqref{eq:lippmann-schwinger} with additional terms. The iterative procedure to solve the integral equation then gives the same regularity properties of $\mathbold{\Phi}$ as in the one channel case.

To retain clarity we demonstrate the bound using the more restrictive assumptions for the off-diagonal potential terms, keeping in mind that a similar bound can be found in the more liberal case of, for instance, the potentials studied in the numerical examples in section~\ref{sec:results}.

\begin{lemma}
\label{lem:integralbound}
	Let $\mathbold{\Phi}$ be the regular solution of eq.~\eqref{eq:matrixeq} where $\mathbold{U} = 2\mu\mathbold{V}$ the potential matrix satisfies the (restrictive) short range conditions
	\begin{equation}
		\label{eq:192007062010}
		\begin{cases}
			\int_0^\infty dr \, r^{l_{i}+l_{j}+2}|V_{ij}(r)| < \infty & \forall i,j\\
			\int_{0}^{\infty} dr\, r^{-l_{i}+l_{j}+1}|V_{ij}(r)| < \infty & \forall i,j,
		\end{cases}
	\end{equation}
	and let $\hat{\mathbold{j}}(kr)$ be the diagonal matrix of spherical Riccati-Bessel functions with corresponding $l_{i}$, then the integral
	\begin{equation}
		\mathbold{\mathcal{I}} = \int_0^\infty dr \, \hat{\mathbold{j}}(kr)\mathbold{U}(r) \mathbold{\Phi}(r) \label{eq:integral_riccati_potential_regular},
	\end{equation}
	is bounded.
\end{lemma}

\begin{proof}
Because of the particular boundary conditions for $r \rightarrow 0$ and the restrictions on the potential matrix $\mathbold{V}$ the regular solution $\mathbold{\Phi}$ fits the matrix integral equation (20.19) in \cite{Taylor2006}
\begin{equation}
	\mathbold{\Phi}(r) = \hat{\mathbold{j}}(kr) + \int_0^r dr^\prime \, \mathbold{g}(r,r^\prime) \mathbold{U}(r^\prime) \mathbold{\Phi}(r^\prime), 
\end{equation}
where $\mathbold{g}(r,r^\prime)$ is a diagonal matrix composed of Green's functions
\begin{equation}
	g_{l_i}(r,r^\prime) = (1/k) \left( \hat j_{l_{i}}(kr)\hat n_{l_{i}}(kr^\prime) -\hat n_{l_{i}}(kr) \hat j_{l_{i}}(kr^\prime)\right),
\end{equation}
on the diagonal.

This regular solution matrix can be written as a series 
$\mathbold{\Phi}(r) = \sum_n \mathbold{\Phi}^{(n)}(r)$ with  the first term $\mathbold{\Phi}^{(0)} = \hat{\mathbold{j}}(kr)$ and the following recurrence relation between $\mathbold{\Phi}^{(n)}$ and $\mathbold{\Phi}^{(n-1)}$ 
\begin{equation}
	\mathbold{\Phi}^{(n)}(r) = \int_0^r dr^\prime \, \mathbold{g}(r,r^\prime)\mathbold{U}(r^\prime) \mathbold{\Phi}^{(n-1)}(r^\prime).
\end{equation}

The integral \eqref{eq:integral_riccati_potential_regular} can then be written as
\begin{align}
	\mathcal{I} &= \sum_{n=0}^{\infty} \int_0^\infty dr \, \hat{\mathbold{j}}(kr) \mathbold{U}(r) \mathbold{\Phi}^{(n)}(r) \\
	&= \int_0^\infty dr \, \hat{\mathbold{j}}(kr) \mathbold{U}(r) \hat{\mathbold{j}}(kr) + \int_0^\infty dr \, \hat{\mathbold{j}}(kr) \mathbold{U}(r) \int_0^r dr_{1} \, \mathbold{g}(r,r_{1}) \mathbold{U}(r_{1})\hat{\mathbold{j}}(kr_{1}) \\
	&+ \cdots + \int_0^\infty dr \, \hat{\mathbold{j}}(kr) \mathbold{U}(r) \int_0^r dr_1 \, \mathbold{g}(r,r_1) \mathbold{U}(r_1) \int_0^{r_1} dr_2 \, \cdots \int_0^{r_{n-1}} dr_n \, \mathbold{g}(r_{n-1},r_n)\mathbold{U}(r_n) \hat{\mathbold{j}}(kr_n) + \cdots
\end{align}

Without loss of generality, let us look for example at a $2\times2$ coupled channel problem with $l_1$ in the first channel and $l_2$ in the second channel. The integral above is then
\begin{align}
	\mathcal{I} &= \int_{0}^{\infty} dr \,
	\begin{pmatrix}
		\hat{j}_{l_1}(kr) & 0 \\
		0 & \hat{j}_{l_2}(kr)
	\end{pmatrix}
	\begin{pmatrix}
		U_{11}(r) &U_{12}(r) \\
		U_{21}(r) &U_{22}(r)
	\end{pmatrix}
	\begin{pmatrix}
		\hat{j}_{l_1}(kr) & 0 \\
		0 & \hat{j}_{l_2}(kr)
	\end{pmatrix} \\
	&+\int_{0}^{\infty} dr \,
	\begin{pmatrix}
		\hat{j}_{l_1}(kr) & 0 \\
		0 & \hat{j}_{l_2}(kr)
	\end{pmatrix}
	\begin{pmatrix}
		U_{11}(r) & U_{12}(r) \\
		U_{21}(r) & U_{22}(r)
	\end{pmatrix} \\
	&\times \int_0^r dr_1 \,
	\begin{pmatrix}
		g_{l_1}(r,r_1) & 0 \\
		0 & g_{l_2}(r,r_1)
	\end{pmatrix}
	\begin{pmatrix}
		U_{11}(r_1) &U_{12}(r_1) \\
		U_{21}(r_1) &U_{22}(r_1)
	\end{pmatrix}
	\begin{pmatrix}
		\hat{j}_{l_1}(kr_1) & 0 \\
		0 & \hat{j}_{l_2}(kr_1)
	\end{pmatrix} + \ldots .
\end{align}

When we look at a particular element of this two by two example we see the contributions
\begin{multline}
	\mathcal{I}_{11} = \int_0^\infty dr \, \hat{j}_{l_1}(kr) U_{11}(r) \hat{j}_{l_1}(kr)
	+ \int_0^\infty dr \, \hat{j}_{l_1}(kr) U_{11}(r) \int_0^r dr_1 \, g_{l_1}(r,r_1) U_{11}(r_1) \hat{j}_{l_1}(kr_1) \\
	+ \int_0^\infty dr \, \hat{j}_{l_1}(kr) U_{12}(r) \int_0^r dr_1 \, g_{l_2}(r,r_1) U_{21}(r_1) \hat{j}_{l_1}(kr_1) + \ldots .
\end{multline}
We note that Riccati-Bessel functions that appear in the integral both have $l=l_1$. In a similar way, 
the matrix element $\mathcal{I}_{12}$ will have the $\hat{j}_{l_1}$ in the beginning of the integral and $\hat j_{l_2}$ at the end.

To show that this integral $\mathcal{I}_{11}$ is finite we make use of the bounds fond in \cite{Newton1982,Newton1960,Taylor2006}. There exist constants $C_{1}$ and $C_{2}$ such that
\begin{align}
	|\hat{j}_l(kr)| &\le C_1 e^{|\text{Im}(k)|r}\left(\frac{|k|r}{1+|k|r}\right)^{l+1} \\
	|g_l(r,r^\prime)| &\le C_2 e^{|\text{Im}(k)|(r-r^\prime)}\frac{1}{|k|} \left(\frac{|k|r}{1+|k|r} \right)^{l+1}\left(\frac{|k|r^\prime}{1+|k|r^\prime} \right)^{-l}.
	\label{eq:bounds}
\end{align}
After substitution we find
\begin{equation}
	\begin{aligned}
		|\mathcal{I}_{11}| \le k^{2l_1+2} &\bigg( C_1^2 \int_0^\infty dr \, e^{2|\text{Im}(k)|r}\left(\frac{r}{1+|k|r}\right)^{2l_1+2}|U_{11}(r)| \\
		& + C_1^2 \int_0^\infty dr \, e^{|\text{Im}(k)|r}\left(\frac{r}{1+|k|r}\right)^{l_1+1} |U_{11}(r) |
		C_2 e^{|\text{Im}(k)|r}\left(\frac{r}{1+|k|r}\right)^{l_1+1} \\
		& \times \int_0^r dr_1 \, C_2 e^{-|\text{Im}(k)|r_1}\left(\frac{r_1}{1+|k|r_1} \right)^{-l_1} |U_{11}(r_1)| e^{|\text{Im}(k)|r_1}\left(\frac{r_1}{1+|k|r}\right)^{l_1+1} \\
		&+ C_1^2 \int_0^\infty dr \, e^{ |\text{Im}(k)|r}\left(\frac{r}{1+|k|r}\right)^{l_1+1} |U_{12}(r) |
		C_2 e^{|\text{Im}(k)|r}\left(\frac{r}{1+|k|r}\right)^{l_2+1} \\
		& \times \int_0^r dr_1 \, C_2 e^{-|\text{Im}|(k)r_1}\left(\frac{r_1}{1+|k|r_1} \right)^{-l_2} |U_{21}(r_1)| e^{|\text{Im}(k)|r_1}\left(\frac{r_1}{1+|k|r}\right)^{l_1+1} \\
		& + \ldots \bigg) \\
		= k^{2l_1+2}& C_{11}, \qquad \text{for some constant } C_{11}.
	\end{aligned}
\end{equation}
This series can be summed up because the potentials $V_{11}$, $V_{12}$, $V_{21}$ and $V_{22}$ satisfy the regularity condition~\eqref{eq:192007062010}. In a similar way we can find bounds for the other elements of $\mathcal{I}$.

In a more general case we have the bound
\begin{equation}
	 \begin{pmatrix}
	 	 |\mathcal{I}_{11}| & |\mathcal{I}_{12}| & \ldots \\ 
		|\mathcal{I}_{21}| & |\mathcal{I}_{22}| &\\
		\vdots  &  & \ddots\\
	 \end{pmatrix}
	\le 
	\begin{pmatrix}
    		 k^{2l_1+2} C_{11} &  k^{l_1 + l_2 +2} C_{12} & \ldots \\
  		k^{l_2 + l_1 +2} C_{21} &  k^{2 l_2 +2} C_{22} &  \\
  		\vdots & & \ddots
   	 \end{pmatrix}.
\end{equation}
This indeed means that the only way expression~\eqref{appeq:det_multi_channel_regularized} can have a zero is through a zero of the Jost function, which ensures that no additional false ``resonances'' or ``bound states'' are introduced by using zeros of this function to locate them. Naturally every zero of $\mathbold{f}(k,\lambda)$ will result in a zero of $\det(\mathbold{F}(k,\lambda))$ which means we are able to track all resonances and bound states through the use of \eqref{appeq:det_multi_channel_regularized}. This proves the equivalence of bound and resonant states and the zeros of $\det(\mathbold{F}(k,\lambda))$.
\end{proof}


\bibliographystyle{elsarticle-num}
\bibliography{coupled}


\end{document}